\documentclass[format=acmsmall, review=false, screen=true]{acmart}

\usepackage{booktabs} 

\usepackage[ruled]{algorithm2e} 
\usepackage{amssymb}
\usepackage{amsmath}
\usepackage{wasysym}
\usepackage{graphicx}
\usepackage{amsfonts}
\usepackage{bbm}
\usepackage{xspace}

\SetAlFnt{\small}
\SetAlCapFnt{\small}
\SetAlCapNameFnt{\small}
\SetAlCapHSkip{0pt}
\IncMargin{-\parindent}

\newcommand{\tup}[1]{\langle #1 \rangle}
\renewcommand{\phi}{\varphi}
\newcommand{\set}[1]{\left\{#1\right\}}
\newcommand{\NN}{{\mathbb{N}}}
\newcommand{\ZZ}{{\mathbb{Z}}}
\newcommand{\RR}{{\mathbb{R}}}
\newcommand{\CC}{{\mathbb{C}}}
\renewcommand{\SS}{{\mathbb{S}}}
\newcommand{\QQ}{{\mathbb{Q}}}
\renewcommand{\C}{\mathcal{C}}
\newcommand{\F}{\mathcal{F}}
\newcommand{\Orb}{\mathcal{O}}
\newcommand{\TT}{\mathbb{T}}
\newcommand{\re}{{\rm Re}}
\newcommand{\im}{{\rm Im}}
\newcommand{\diagM}{{\rm diag}}
\newcommand{\ray}{\mathsf{r}}
\newcommand{\norm}[1]{{\left\lVert#1\right\rVert}}

\newcommand{\tp}{\intercal}

\newcommand{\newextmathcommand}[2]{%
	\newcommand{#1}{\ensuremath{#2}\xspace}
}
\newextmathcommand{\KK}{\mathbbm{K}}
\newcommand{\sset}{\subseteq}

\newextmathcommand{\Inv}{\mathcal{I}}

\newextmathcommand{\theo}{\mathfrak{R}}
\newextmathcommand{\theosmaller}{\mathfrak{R}_{1}}
\newextmathcommand{\theoreals}{\mathfrak{R}_{0}}
\newextmathcommand{\theoexp}{\mathfrak{R}_{\exp}}
\newextmathcommand{\theopow}{\mathfrak{R}_{\rm pow}}

\newcommand{\vectau}{\boldsymbol\tau}
\renewcommand{\Re}{\re\,}
\renewcommand{\Im}{\im\,}

\newcommand{\lowervec}{\boldsymbol{\ell}}
\newcommand{\uppervec}{\mathbf{u}}
\newcommand{\boxvec}{{\rm Box}}

\newtheorem{remark}{Remark}
\newtheorem{claim}{Claim}

\begin{document}
\title[O-Minimal Invariants for Discrete-Time Dynamical
Systems]{O-Minimal Invariants for Discrete-Time Dynamical Systems}

\author{Shaull Almagor}
\orcid{1234-5678-9012-3456}
\affiliation{%
	\institution{Computer Science Department, Technion}
	\streetaddress{Taub Building}
	\city{Haifa}
	\postcode{3200003}
	\country{Israel}}
\email{shaull@cs.technion.ac.il}
\author{Dmitry Chistikov}
\affiliation{%
  \institution{Centre for Discrete Mathematics and its
  	Applications (DIMAP) \&
  	Department of Computer Science, University of Warwick}
  \streetaddress{}
  \city{Coventry, CV4 7AL}
  \country{United Kingdom}
}
\email{d.chistikov@warwick.ac.uk}
\author{Jo\"el Ouaknine}
\affiliation{%
 \institution{Max Planck Institute for Software Systems}
 \streetaddress{Saarland Informatics Campus}
 \city{Saarbr\"ucken}
 \country{Germany}}
\affiliation{%
	\institution{Department of Computer Science, Oxford University}
	\streetaddress{Parks Road}
	\city{Oxford}
	\postcode{OX1 3QD}
	\country{United Kingdom}}
\email{joel@mpi-sws.org}
\author{James Worrell}
\affiliation{%
	\institution{Department of Computer Science, Oxford University}
	\streetaddress{Parks Road}
	\city{Oxford}
	\postcode{OX1 3QD}
	\country{United Kingdom}}
\email{jbw@cs.ox.ac.uk}

\begin{abstract}
Termination analysis of linear loops plays a key r\^{o}le in several
areas of computer science, including program verification and abstract
interpretation.   Already for the simplest variants of linear loops the question of termination 
relates to deep open problems in number theory, such as the decidability of the Skolem and Positivity
Problems for linear recurrence sequences, or equivalently reachability
questions for discrete-time linear dynamical systems. In this paper,
we introduce the class of \emph{o-minimal invariants}, which is
broader than any previously considered, and study the decidability of
the existence and algorithmic synthesis of such invariants as
certificates of non-termination for linear loops equipped with a large
class of halting conditions. We establish two main decidability
results, one of them conditional on Schanuel's conjecture in
transcendental number theory.
\end{abstract}


%
%


\maketitle


\section{Introduction}
\label{sec:intro}

This paper is concerned with the existence and algorithmic
synthesis of suitable \emph{invariants} for discrete-time linear dynamical
systems. Invariants are one of the most fundamental and useful
notions in the quantitative sciences, and within computer
science play a central r\^ole in areas such as program
analysis and verification, abstract interpretation, static
analysis, and theorem proving. To this day, automated
invariant synthesis remains a topic of active research; see,
e.g., \cite{KCBR18}, and particularly Sec.~8 therein.

In program analysis, invariants are often invaluable tools enabling
one to establish various properties of interest. Our focus here is on simple
linear loops, of the following form:
\begin{equation}
	\label{loop}
	P : \ \ x \leftarrow s;\; \ \texttt{while}\ x \notin F\ \texttt{do}\ x
	\leftarrow A x \, , 
\end{equation}
where $x$ is a $d$-dimensional column vector of variables, $s$ is a $d$-dimensional
vector of integer, rational, or real numbers, $A \in \mathbb{Q}^{d \times d}$ is a square rational matrix
of dimension $d$, and $F \subseteq \mathbb{R}^d$ represents the
halting condition.

Much research has been devoted to the termination analysis of such
loops (and variants thereof); see, e.g.,~\cite{BGM12,BG14,OW15}. For
$S \subseteq \mathbb{R}^d$, we say that $P$ \emph{terminates} on $S$
if it terminates for all initial vectors $s \in S$. One of the
earliest and most famous results in this line of work is due to Kannan
and Lipton, who showed polynomial-time decidability of termination in
the case where $S$ and $F$ are both singleton vectors with rational
entries~\cite{KL80,KL86}. This work was subsequently extended to
instances in which $F$ is a low-dimensional vector
space~\cite{COW13,COW16} or a low-dimensional
polyhedron~\cite{COW15}. Still starting from a fixed initial vector,
the case in which the halting set $F$ is a hyperplane is equivalent to
the famous Skolem Problem for linear recurrence sequences, whose
decidability has been open for many decades~\cite[\S 3.9]{Tao08},
although once again positive results are known in low
dimensions~\cite{MST84,Ver85}. The case in which $F$ is a half-space
corresponds to the Positivity Problem for linear recurrence sequences,
likewise famously open in general but for which some partial results
also exist~\cite{OW14a,OW14b}.

Cases in which the starting set $S$ is infinite have also been
extensively studied, usually in conjunction with a halting set $F$
consisting of a half-space. For example, decidability of termination
for $S = \mathbb{R}^d$, $S = \mathbb{Q}^d$, and $S = \mathbb{Z}^d$ are
known~\cite{Tiw04,Bra06,OPW15,HOW19}. In the vast majority
of cases, however, termination is a hard problem
(and often undecidable~\cite{XZ10}), which has
led researchers to turn to semi-algorithms and heuristics. One of
the most popular and successful approaches to establishing termination is the use
of ranking functions, on which there is a substantial body of work;
see, e.g.,~\cite{BG14}, which includes a broad survey on the
subject.

Observe, for a loop $P$ such as that given in~(\ref{loop}), that
failure to terminate on a set $S$ corresponds to the existence of some
vector $s \in S$ from which $P$ loops forever. It is important to
note, however, that the absence of a suitable ranking function does
not necessarily entail non-termination, owing to the non-completeness
of the method. Yet surprisingly, as pointed out in~\cite{GHMRX08},
there has been significantly less research in methods seeking to
establish \emph{non-termination} than in methods aimed at proving
termination. Most existing efforts for the former have focused on the
synthesis of appropriate invariants; see,
e.g.,~\cite{CH78,CSS03,SSM04,RK04,Cou05,RK07,FOOPW17,FOOPW19,FLOOPW19}.

In order to make this notion more precise, let us associate with our
loop $P$ a \emph{discrete-time linear dynamical system} $(A,s)$. The
\emph{orbit} of this dynamical system is the set 
$\mathcal{O} = \{A^ns \mid n \geq 0\}$. It is clear that $P$ fails to terminate
from $s$ iff $\mathcal{O}$ is disjoint from $F$. A possible method to
establish the latter is therefore to exhibit a set $\mathcal{I} \subseteq
\mathbb{R}^d$ such that:
\begin{enumerate}
	\item $\mathcal{I}$ contains the initial vector $s$, i.e., $s \in \mathcal{I}$;
	\item $\mathcal{I}$ is invariant under $A$, i.e., $A \mathcal{I}
	\subseteq \mathcal{I}$; and
	\item $\mathcal{I}$ is disjoint from $F$, i.e., $\mathcal{I}
	\cap F = \emptyset$.
\end{enumerate}
Indeed, the first two conditions ensure that $\mathcal{I}$ contains
the entire orbit $\mathcal{O}$, from which the desired claim follows
thanks to the third condition.

In instances of non-termination, one notes that the orbit
$\mathcal{O}$ itself is always an invariant meeting the above
conditions. However, since in general one does not know how to
algorithmically check Condition~(3), such an invariant is of little
use. One therefore usually first fixes a suitable class of candidate sets for
which the above conditions can be mechanically verified, and within
that class, one seeks to determine if an invariant can be
found. Examples of such classes include
polyhedra~\cite{CH78}, algebraic sets~\cite{RK07}, and semi-algebraic
sets~\cite{FOOPW17}.

\subparagraph*{\textbf{Main contributions.}}
We focus on loops of the form given in~(\ref{loop}) above.  We
introduce the class of \emph{o-minimal invariants}, which, to the best
of our knowledge, is significantly broader than any of the classes
previously considered in the context of linear loops.  An o-minimal
invariant is one that is definable in some o-minimal expansion 
of the ordered field $\theoexp$ of real numbers with real
exponentiation.  We also consider two large classes of halting sets,
namely those definable over the ordered field $\theoreals$ of real
numbers (i.e., \emph{semi-algebraic sets}) and those definable
in $\theoexp$.

Given $s \in \mathbb{Q}^d$, $A \in \mathbb{Q}^{d\times d}$, and $F
\subseteq \mathbb{R}^d$, our main results are the following: if $F$ is
a semi-algebraic set, it is decidable whether there exists an
o-minimal invariant $\mathcal{I}$ containing $s$ and disjoint from
$F$, and moreover in positive instances such an invariant can be
defined explicitly in \theoreals.  For the more general case in which
$F$ is $\theoexp$-definable, assuming Schanuel's conjecture it is
decidable whether there exists an o-minimal invariant $\mathcal{I}$
containing $s$ and disjoint from $F$, and moreover in positive
instances such an invariant can be defined explicitly in \theoexp.

We illustrate below some of the key ideas from our
approach. Consider a linear dynamical system $(A,s)$ with $A\in \QQ^{3\times 3}$ whose
orbit $\Orb$ is depicted in Figure~\ref{fig:orbit}. In our example,
$\Orb$ spirals outward at some rate $\rho_1$ in the $x,y$-plane,
and increases along the $z$-axis at some rate
$\rho_2$. Intuitively, $\rho_1$ and $\rho_2$ are the moduli of the
eigenvalues of $A$.

We now consider a `normalised' version of $A$, with both moduli set
to $1$. We then connect every point on the normalised orbit with a
\emph{trajectory ray} to its corresponding point on $\Orb$, while
respecting the rates $\rho_1$ and $\rho_2$ (see
Figure~\ref{fig:rays}). One can observe that the normalised orbit is
dense in the unit circle. We prove that \emph{any} o-minimal invariant for
$(A,s)$ must in fact eventually contain every trajectory ray for
every point on the unit circle; we depict 
the union of these rays,
referred to as the \emph{trajectory cone}, in
Figure~\ref{fig:cone}. Finally, we show that any o-minimal invariant
must in fact contain some truncation of the trajectory cone from below, starting from some
height. That is, there is a uniform bound from which all the rays
must belong to the invariant. Moreover, we can now synthesise an
$\theoexp$-definable o-minimal invariant by simply adjoining a
finite number of orbit points to the truncated trajectory cone, as
depicted in Figure~\ref{fig:conetail}.

\begin{figure}[htbp]
	\begin{minipage}[b]{0.22\linewidth}
		\centering
		\includegraphics[width=\linewidth]{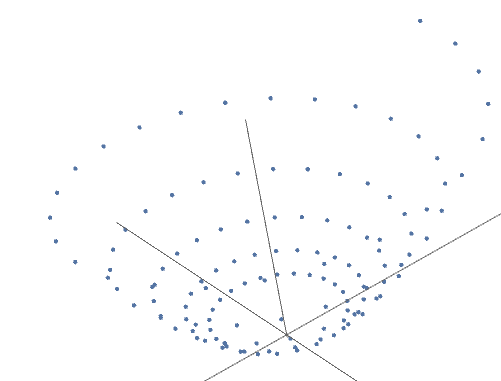}
		\caption{The orbit $\Orb$ of $(A,s)$}
		\label{fig:orbit}
	\end{minipage}
	\hspace{0.02\linewidth}
	\begin{minipage}[b]{0.22\linewidth}
		\centering
		\includegraphics[width=\linewidth]{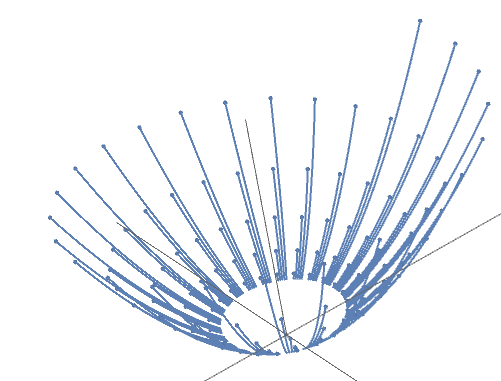}
		\caption{Trajectory rays of $\Orb$.}
		\label{fig:rays}
	\end{minipage}
	\hspace{0.02\linewidth}
	\begin{minipage}[b]{0.22\linewidth}
		\centering
		\includegraphics[width=\linewidth]{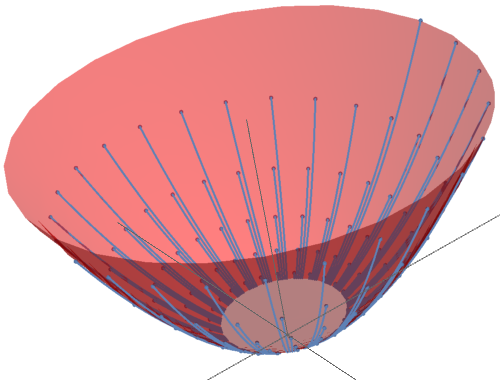}
		\caption{Trajectory cone for $\Orb$.}
		\label{fig:cone}
	\end{minipage}
	\hspace{0.02\linewidth}	
	\begin{minipage}[b]{0.22\linewidth}
		\centering
		\includegraphics[width=\linewidth]{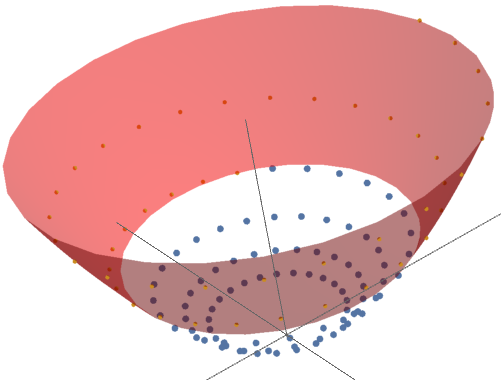}
		\caption{Invariant set for $\Orb$.}
		\label{fig:conetail}
	\end{minipage}
\end{figure}

It is worth emphasising that, whilst in general there cannot exist a smallest o-minimal
invariant, the family of truncated cones that we define plays the
r\^{o}le of a `minimal class', in the sense that \emph{any}
o-minimal invariant must necessarily contain some truncated cone. We make all of
these notions precise in the main body of the paper.

The works that are closest to ours in the literature
are~\cite{FOOPW17}, \cite{FOOPW19}, and \cite{FLOOPW19}, which
consider the same kind of loops as we do here, but restricted to the
case in which the halting set $F$ is always a rational singleton. The
authors then exhibit procedures for deciding the existence of
semi-algebraic invariants (\cite{FOOPW17,FOOPW19}) and semi-linear
invariants~\cite{FLOOPW19}. The present paper has a considerably
broader scope, in that we deal with much wider classes both of
invariants and halting sets. From a technical standpoint, the present
paper correspondingly makes heavy use of model-theoretic and
number-theoretic tools that are entirely absent from the above papers.

\section{Preliminaries and Main Definitions}
\label{sec:prelim}

We write $\theoreals$ for the structure $\tup{\RR,0,1,+,\cdot,<}$,
i.e., the ordered field of real numbers with constants $0$ and $1$.  A
sentence in the corresponding first-order language can be considered
as a quantified Boolean combination of atomic propositions of the form
$P(x_1,\ldots,x_n)>0$, where $P$ is a polynomial with integer
coefficients and $x_1,\ldots,x_n$ are variables. Tarski famously
showed that the first-order theory of $\theoreals$ admits quantifier
elimination~\cite{tarski1951decision} and is therefore decidable.  In
addition to $\theoreals$, we also consider the structure $\theoexp$,
obtained by expanding $\theoreals$ with the real exponentiation
function $x \mapsto e^x$.  It is an open question whether the theory
of the reals with exponentiation is decidable; however decidability
was established subject to Schanuel's conjecture by MacIntyre and
Wilkie~\cite{MacintyreWilkie1996}.  (Schanuel's conjecture is a
unifying conjecture in transcendental number theory that generalises
many of the classical results of that subject.)  MacIntyre and Wilkie
further showed in~\cite{MacintyreWilkie1996} that decidability of the
theory of $\theoexp$ implies a weak form of Schanuel's conjecture.

Let $\theo$ be an expansion of the structure $\theoreals$.
A set $S\subseteq \RR^d$ is \emph{definable} in $\theo$ if
there exists a formula $\phi(x_1,\ldots,x_d)$ in $\theo$ with free
variables $x_1,\ldots,x_d$ such that $S=\set{(c_1,\ldots,c_d)\in \RR^d
  \mid \theo \models \phi(c_1,\ldots,c_d)}$. A function
$f\colon B\to \RR^m$ with $B\subseteq \RR^n$ is \emph{definable} in
$\theo$ if its graph $\Gamma(f)=\set{(x,f(x)) \mid x\in B}\subseteq
\RR^{n+m}$ is an $\theo$-definable set.
For $\theo = \theoreals$, the ordered field of real numbers,
$\theoreals$-definable sets (resp.\ functions) are known as
\emph{semi-algebraic} sets (resp.\ functions).

\begin{remark}
	Our usage of the terms ``definable'' and ``semi-algebraic''
	corresponds to ``definable without parameters'' and
	``semi-algebraic without parameters'' in model theory.
\end{remark}

\begin{remark}
	\label{rmk:complex definable}
        Recall that there is a natural first-order interpretation of
        the field of complex numbers $\CC$ in the field of real
        numbers $\RR$.  We shall say that a set $S\subseteq \CC^d$ is
        \emph{$\theo$-definable} if the image $\{ (x, y) \in \RR^d
        \times \RR^d \mid x + i y \in S \}$ of $S$ under this
        interpretation is $\theo$-definable.
\end{remark}

A totally ordered structure $\tup{M,<,\ldots}$ is said to
be \emph{o-minimal} if every definable subset of $M$ is a finite union
of intervals.  Tarski's result on quantifier
elimination~\cite{tarski1951decision} implies that $\theoreals$ is
o-minimal. The o-minimality of $\theoexp$ is due to
Wilkie~\cite{Wilkie96} and holds unconditionally.  An o-minimal
expansion $\theo$ of $\theoreals$ satisfies the following useful
properties (see~\cite{dries_1998} for precise definitions and proofs).
\begin{enumerate}
	\item
	For an $\theo$-definable set $S\subseteq \RR^d$, its topological closure $\overline{S}$ is also $\theo$-definable.
	\item
	For an $\theo$-definable function $f \colon S\to \RR$, the number $\inf\set{f(x) \mid x\in S}$ is $\theo$-definable (as a singleton set).
	\item
	O-minimal structures admit \emph{cell decomposition}: every $\theo$-definable set $S\subseteq \RR^d$ can be written as a finite union of connected components called \emph{cells}. Moreover, each cell is $\theo$-definable and homeomorphic to $(0,1)^m$ for some $m \in \{0, 1, \ldots, d\}$ (where for $m=0$ we have that $(0,1)^0$ is a single point, namely $\{\vec{0}\}\subseteq \RR^d$). The \emph{dimension} of $S$ is defined as the maximal such $m$ occurring in the cell decomposition of $S$. 
	\item
	For an $\theo$-definable function $f \colon S\to \RR^m$, the dimension of its graph $\Gamma(f)$ is the same as the dimension of $S$.
\end{enumerate}

As mentioned above, $\theoreals$ is decidable thanks to its effective
quantifier elimination procedure. Equivalently, given a semi-algebraic
set, we can effectively compute its cell decomposition.
Unfortunately, few more expressive theories are known to be
unconditionally decidable.  Our decidability result in
Theorem~\ref{thm: theo exp decidable schanuel} on invariants definable
in $\theoexp$ is subject to Schanuel's conjecture; somewhat
surprisingly, however, we exhibit in Theorem~\ref{thm: qsa invariants
  against sa decidable} an unconditional decidability result.

A \emph{discrete-time linear dynamical system} (LDS) consists of a
pair $(A,s)$, where $A \in \mathbb{Q}^{d \times d}$ and $s \in
\mathbb{Q}^d$. Its \emph{orbit} $\Orb$ is the set $\set{A^n s \mid
  n\in \NN}$. An \emph{invariant} for $(A,s)$ is a set $\mathcal{I}
\subseteq \mathbb{R}^d$ that contains $s$ and is stable under
applications of $A$, i.e., $A \mathcal{I} \subseteq
\mathcal{I}$. Given a set $F \subseteq \mathbb{R}^d$, we say that the
invariant $\mathcal{I}$ \emph{avoids} $F$ if the two sets are
disjoint.  An \emph{o-minimal invariant} is one that is definable in
an o-minimal expansion of $\theoexp$.

\section{From the Orbit to Trajectory Cones and Rays}
\label{sec: traj cones}
Let $(A, s)$ be an LDS with $A\in \QQ^{d\times d}$ and $s\in
\QQ^d$. We consider the orbit $\Orb=\set{A^ns \mid n\in \NN}$.
Write $A$ in Jordan form as $A=PJP^{-1}$ where 
$P$ is an invertible matrix, and 
$J$ is a block diagonal matrix of the form $J=\diagM(B_1,\ldots,B_k)$, where for every $1\le i\le k$, $B_i\in \CC^{d_i\times d_i}$ is a Jordan block corresponding to an eigenvalue $\lambda_i$:
{
	\setlength{\arraycolsep}{2pt}
	$$B_i=\begin{pmatrix}
	\lambda_i & 1 & & \\[-6pt]
	& \ddots  & \ddots & \\[-6pt]
	&  &\ddots  & 1\\[-6pt]
	& &  &\lambda_i\\
	\end{pmatrix} \, .$$
}%
To reflect the
block structure of $J$, we often range over $\set{1,\ldots,d}$ via a
pair $(i,j)$, with $1\le i\le k$ and $1\le j\le d_i$, which denotes
the index corresponding to row $j$ in block $i$; we refer to this
notation as \emph{block-row indexing}.

	Henceforth, we assume that for all $1\le i\le k$ we have that
        $\lambda_i\neq 0$ (i.e., that the matrices $A$ and $J$ are
        invertible). Indeed, if $\lambda_i=0$, then $B_i$ is a
        nilpotent block and therefore, for the purpose of invariant
        synthesis, we can ignore finitely many points of the orbit
        under $A$ until $B^n_i$ is the $0$ block. We can then restrict
        our attention to the image of $A^n$, by identifying it with
        $\RR^{d-d_i}$.  

For all $i\in\{1,\ldots,k\}$ we can write $\lambda_i=\rho_i\xi_i$
where $\rho_i>0$ is positive real and $\xi_i$ is a complex number of
absolute value~1, with both $\rho_i$ and $\xi_i$ being algebraic.

Observe that now, for every set $F \subseteq \RR^d$, we have that $A^ns\in F$ iff $J^n s'\in P^{-1}F$ where $s'=P^{-1}s$.
For every $n>d$, $J^n=\diagM(B_1^n,\ldots,B_k^n)$ with 
$$B_i^n=\begin{pmatrix}
\lambda_i^n & \frac{n}{\lambda_i}\lambda_i^n & \cdots & \frac{{n\choose d_i-1}}{\lambda_i^{d_i-1}}\lambda_i^n\\
& \ddots &   & \vdots \\
& & &\lambda_i^n \\
\end{pmatrix} \, .$$
Every coordinate of $J^n s'$ is of the form $\lambda_i^n Q_{i,j}(n)=
(\rho_i\xi_i)^n Q_{i,j}(n)=\rho_i^n\xi_i^n Q_{i,j}(n)$ for some $1\le
i\le k$ and $1\le j\le d_i$, where $Q_{i,j}$ is a polynomial (possibly
with complex coefficients) that depends on $J$ and $s'$.  

Let $R=\diagM(\rho_1,\ldots,\rho_k)$ and
$L=\diagM(\xi_1,\ldots,\xi_k)$. We define $\TT$ to be the
subgroup of the torus in $\CC^k$ generated by the multiplicative
relations of the normalised eigenvalues
$\xi_1,\ldots,\xi_k$. That is, consider the subgroup $G=\set{v
	= (v_1, \ldots, v_k)\in \ZZ^k \mid \xi_1^{v_1}\cdots \xi_k^{v_k}=1}$ of $\ZZ^k$, and let
$$\TT=\set{(\alpha_1, \ldots, \alpha_k) \in \CC^k \mid |\alpha_i|=1
  \text{ for all $i$, and for every }v\in G,\ \alpha_1^{v_1}\cdots
  \alpha_k^{v_k}=1}.$$ A result by Masser~\cite{Mas88} allows to
compute a basis for $G$, and hence a representation of $\TT$.  Using
Kronecker's theorem on inhomogeneous simultaneous Diophantine
approximation~\cite{cassels1965introduction} it is shown in
\cite{ouaknine2014ultimate} that $\set{L^n \mid n\in \NN}$ is a dense
subset of $\set{{\rm diag}(\alpha_1,\ldots,\alpha_k) \mid
  (\alpha_1,\ldots,\alpha_k)\in \TT}$.

Thus, for every $n\in \NN$, we have
$$J^ns'\in \set{\begin{pmatrix}
	\rho_1^n p_1 Q_{1,1}(n),
	\ldots,		
	\rho_k^n p_k Q_{k,{d_k}}(n)
	\end{pmatrix}^\tp\ \mid \ (p_1,\ldots,p_{k})\in \TT} \, .$$

We now define a continuous over-approximation of the expressions
$\rho_i^n$ by replacing $n \in \mathbb{N}$ with $\log t$, where $t\geq
1$ is a real variable, so that, writing $b_i:=\log \rho_i$,
$\rho_i^n$ becomes $t^{b_i}$.  This over-approximation leads to the
following definition, which is central to our approach.
\begin{definition}
	\label{def:trajectory cone}	
	For every $t_0\ge 1$, we define the \emph{trajectory cone}\footnote{%
		These sets are, of course, not really cones. Nevertheless, if
		for all $i$ we have $b_i = 1$ and the polynomials $Q_{i,j}$ are constant,
		then the set $\C_{t_0}$ is a conical surface
		formed by the union of rays going from the origin through all points
		of $\TT$. The initial segments of the rays, of length determined
		by the parameter~$t_0$, are removed.%
	} for the orbit $\Orb$ as
	$$\C_{t_0}=\set{\begin{pmatrix}
		t^{b_1} p_1 Q_{1,1}(\log t),
		\ldots,
		t^{b_k} p_k Q_{k,d_k}(\log t)
		\end{pmatrix}^\tp\ \mid \ (p_1,\ldots,p_{d})\in \TT,\ t\ge t_0}.$$	
\end{definition}
In particular, we have that $J^n s'\in \C_1$.

In order to analyse invariants, we require a finer-grained notion than the entire trajectory cone. 
To this end, we introduce the following.
\begin{definition}
	\label{def: ray}
	For every 
	$p=(p_1,\ldots, p_k)\in \TT$ and every $t_0 \geq 1$, we define the \emph{(trajectory) ray}%
	\footnote{%
		Likewise, this set is not, strictly speaking, a straight
		half-line.%
	}
	$
	\ray(p,t_0)=\set{\begin{pmatrix}
		t^{b_1} p_1 Q_{1,1}(\log t), \ldots, 
		t^{b_k} p_k Q_{k,d_k}(\log t)
		\end{pmatrix}^\tp  \mid t\ge t_0}. 
	$
\end{definition}
Observe that we have $\C_{t_0}=\bigcup_{p\in \TT}\ray(p,t_0)$.

\begin{example}
	\label{xmp: 2 5}
	Consider the matrix $A=\diagM(5, 2)$
	and the initial point $s=(1, 1)^\tp$.
	We then have $\TT=\set{\begin{pmatrix}
		1,1
		\end{pmatrix}}$ and $\C_{t_0}=\set{\begin{pmatrix}
		t^{\log 5}, t^{\log 2}
		\end{pmatrix}^\tp \mid t\ge t_0}$. 
Observe that this is not an $\theoreals$-definable set, as the
quotient $\frac{\log 5}{\log 2}$ is not rational. This shows that even
for diagonalizable matrices (where $\C_{t_0}$ has a simple form,
devoid of the polynomials $Q_{i,j}$), $\theoreals$ might not be enough
to recover definability of the orbit (in the sense of
Theorem~\ref{thm: P traj cone is invariant} below).
\end{example}

\section{Constructing Invariants from Trajectory Cones}
\label{sec: cones are invariants}

We now proceed to show that the trajectory cones defined in
Section~\ref{sec: traj cones} can be used to characterise o-minimal
invariants.
More precisely, we show that for an LDS $(A, s)$ with $A=PJP^{-1}$,
the image under $P$ of every trajectory cone $\C_{t_0}$, augmented
with finitely many points from $\Orb$, is an invariant. Moreover, we
show that such invariants are $\theoexp$-definable, and hence
o-minimal. Complementing this, we show in Section~\ref{sec: minimal} that \emph{every}
o-minimal invariant must contain some trajectory cone.

In what
follows, let $A = PJP^{-1}$, $s$, as well as the real numbers $b_1, \ldots, b_d$ be defined
as in Section~\ref{sec: traj cones}.

\begin{theorem}
	\label{thm: P traj cone is invariant}
	For every $t_0\ge 1$, the set $P\cdot \C_{t_0}\cup \set{A^n s \mid n< \log t_0}$ is an $\theoexp$-definable invariant for the LDS $(A, s)$.
\end{theorem}

The intuition behind Theorem~\ref{thm: P traj cone is invariant} is as
follows.  Clearly, the orbit $\Orb$ itself is always an invariant for
$(A, s)$. However, it is generally not definable in any o-minimal
structure (in particular, since it has infinitely many connected
components).  In order to recover definability in $\theoexp$ while
maintaining stability under $A$, the invariants constructed in
Theorem~\ref{thm: P traj cone is invariant} over-approximate the orbit
by the image of the trajectory cone $\C_{t_0}$ under the linear
transformation $P$. Finally, a finite set of points from $\Orb$ is
added to this image of the trajectory cone, to fill in the missing
points in case $t_0$ is too large.

The proof of Theorem~\ref{thm: P traj cone is invariant} has several parts.
First, recall that the trajectory cone itself, $\C_{t_0}$, is an
over-approximation of the set $\set{J^n P^{-1}s \mid n\in \NN}$. As such,
clearly $\C_{t_0} \subseteq \CC^d$. In comparison, the orbit can be
written as $\Orb=\set{P J^n P^{-1}s \mid n\in \NN}\subseteq \RR^d$. We
prove in 
Section~\ref{subsec: P ray is real}
the following simple lemma,
from which it follows that the entire set $P\cdot \C_{t_0}$ is also a subset of $\RR^d$.

\begin{lemma}
	\label{lem: P ray is real}
	For every $p\in \TT$ and $t_0\ge 1$, we have $P\cdot \ray(p,t_0)\subseteq \RR^d$.
\end{lemma}


In the second part of the proof of Theorem~\ref{thm: P traj cone is invariant},
we show that $P\cdot \C_{t_0}$ is stable under $A$.
The key ingredient is the following lemma, which characterises the action of $J$ on rays and is
proved in Section~\ref{subsec: J on ray}.

\begin{lemma}
	\label{lem:J on ray}
	For every $p=(p_1,\ldots,p_k)\in \TT$ and $t_0 \ge 1$, we have $J \cdot \ray(p,t_0)=\ray(L p,\rho_k t_0)$.
\end{lemma}

The next lemma then lifts Lemma~\ref{lem:J on ray} to the entire trajectory cone.

\begin{lemma}
	\label{lem: traj cones are invariant}
	For every $t_0\ge 1$, we have $J\cdot \C_{t_0}\subseteq \C_{t_0}$.
\end{lemma}

\begin{proof}
	Recall that $\C_{t_0}=\bigcup_{p\in \TT}\ray(p,t_0)$. By
	Lemma~\ref{lem:J on ray} we have that
	$J\cdot \C_{t_0}=\bigcup_{p\in \TT}\ray(Lp,e t_0)$.  But $e t_0 > t_0$ and
	$L p\in \TT$ iff
	$p\in \TT$. Hence we have that
	$\ray(Lp,et_0) \subseteq \ray(Lp,t_0)$, from which we
	conclude that
	$J\cdot \C_{t_0}\subseteq \bigcup_{p\in \TT}\ray(Lp,t_0)=\bigcup_{p\in
		\TT}\ray(p,t_0)=\C_{t_0}$.
\end{proof}

The proof of Theorem~\ref{thm: P traj cone is invariant}
combines all these ingredients together and is given in
subsection~\ref{subsec: proof that P traj cone is invariant}.

\subsection{Proof of Lemma~\ref{lem: P ray is real}}
\label{subsec: P ray is real}

Recall that we have $A=PJP^{-1}$, where $J=\diagM(B_1,\ldots,B_k)$ is
a block diagonal matrix with $B_i \in \CC^{d_i\times d_i}$ the Jordan
block corresponding to eigenvalue $\lambda_i$.  Write
$P=\begin{pmatrix} P_1 & \ldots & P_k \end{pmatrix}$, where $P_i \in
\CC^{d\times d_i}$ for all $i\in \{1,\ldots,k\}$.  Since the
generalized eigenspaces of $A$ respectively corresponding to pairs of
complex-conjugate eigenvalues are themselves element-wise conjugate,
we can partition the set $\{1,\ldots,k\}$ into singletons and pairs of
the form $\{i_1,i_2\}$ such that $d_{i_1}=d_{i_2}$,
$P_{i_1}=\overline{P_{i_2}}$, and $B_{i_1}=\overline{B_{i_2}}$.  In
this case we say that $i_1$ and $i_2$ are \emph{conjugate block
	indices}.

By definition, for conjugate block indices $i_1,i_2$ we have that for
all $j\in\{1,\ldots,d_{i_1}\}$ the column of $P$ with block-column
index $(i_1,j)$ is conjugate to the column of $P$ with block-column
index $(i_2,j)$.  Likewise the row of $P^{-1}$ with block-row index
$(i_1,j)$ is conjugate to that with index $(i_2,j)$.\footnote{Since
	$\overline{P}=PS$ for $S$ the permutation matrix that interchanges
	conjugate blocks, we have $\overline{P^{-1}}=S^{-1}P^{-1}$.}  In
particular, for the vector $s'=P^{-1}s$ we have that the entries
$s'_{i_1,j}$ and $s'_{i_2,j}$ are complex conjugates.

Let $p\in \TT$ and $t_0\ge 1$.
Consider the vector $v:=\begin{pmatrix}
t^{b_1} p_1 Q_{1,1}(\log t)\\
\vdots\\
t^{b_k} p_k Q_{k,d_k}(\log t)\\
\end{pmatrix}\in \ray(p,t_0)$.
Fix two conjugate block indices $i_1,i_2 \in \{1,\ldots,k\}$.  We
claim that for all $j\in\{1,\ldots,d_{i_1}$ the entries of $v$ with
respective block-row indices $(i_1,j)$ and $(i_2,j)$, namely
$t^{b_1} p_{i_1} Q_{i_1,j}(\log t)$ and 
$t^{b_2} p_{i_2} Q_{i_2,j}(\log t)$, are mutually conjugate.

Towards proving the claim,
observe that for every $1\le i\le k$ and $1\le j\le d_i$ we have
$$Q_{i,j}(\log t)=\sum_{m=0}^{d_i-j}\frac{{\log t \choose
		m}}{\lambda_i^m} \cdot s'_{i,j+m}$$ with $(i,j+m)$ being a
block-row index.\footnote{Here, for $s\in \RR$ and $m\in \NN$, one
	defines ${s\choose m}=\frac{1}{m!}\prod_{i=0}^{m-1}(s-i)$, which
	maintains consistency with the original definition of $Q_{i,j}$ in
	Section~\ref{sec: traj cones}.}  It follows that the values
$Q_{i_1,j}(\log t)$ and $Q_{i_2,j}(\log t)$ are complex conjugates.
Note moreover that we have $p_{i_1}p_{i_2}=1$ since $p \in \TT$ and
$\xi_{i_1}\xi_{i_2}=1$.  Thus $p_{i_1}$ and $p_{i_2}$ are also complex
conjugates.  Finally we note that $b_{i_1}=\log |\lambda_{i_1}|=\log
|\lambda_{i_2}|= b_{i_2}$ and hence $t^{b_{i_1}}=t^{b_{i_2}}$.  The
claim follows.

Given the above claim and the fact that for conjugate block indices
$i_1$ and $i_2$, for all $j$ the respective rows of $P$ with indices
$(i_1,j)$ and $(i_2,j)$ are element-wise conjugate, we conclude that
$Pv\in\mathbb{R}$.  This concludes the proof.\qed

\subsection{Proof of Lemma~\ref{lem:J on ray}}
\label{subsec: J on ray}

Let $y=\begin{pmatrix}
t^{b_1} p_1 Q_{1,1}(\log t)\\
\vdots\\
t^{b_k} p_k Q_{k,d_k}(\log t)\\
\end{pmatrix}\in \ray(p,t_0)$. We claim that $Jy=\begin{pmatrix}
(e t)^{b_1} \xi_1 p_1 Q_{1,1}(\log (e t))\\
\vdots\\
(e t)^{b_k} \xi_k p_k Q_{k,d_k}(\log (e t))\\
\end{pmatrix}$. Note that since $Lp=(\xi_1
p_1,\ldots,\xi_k p_k)$, the above suffices to conclude the proof.

Consider a coordinate $m=(i,j)$ of $Jy$ in block-row index, with $j<d_i$. 
The case of $j=d_i$ is similar and simpler.  
To simplify notation, we write $\xi,\rho,$ and $d$ instead of $\xi_i,\rho_i,$ and $d_i$, respectively. 
Then we have 
$$(Jy)_m=\xi \rho t^{b_i}p_iQ_{i,j}(\log t)+
t^{b_i}p_iQ_{i,j+1}(\log t) \, .$$
Recall that\footnote{Here, for $w\in \RR$ and $m\in \NN$, one defines
	${w\choose m}=\frac{1}{m!}\prod_{i=0}^{m-1}(w-i)$, which maintains
	consistency with the original definition of $Q_{i,j}$ in Section~\ref{sec: traj
		cones}.} 
$$Q_{i,j}(\log t)=\sum_{c=0}^{d-j}\frac{{\log t
		\choose c}}{(\rho\xi)^c}s'_{i,j+c} \, ,$$ 
with $(i,j+c)$ in block-row index. 
We can then write 
\begin{equation}
	\label{eq: Jy}
	(Jy)_m=\xi \rho
	t^{b_i}p_i\sum_{c=0}^{d-j}\frac{{\log t \choose
			c}}{(\rho\xi)^c}s'_{i,j+c}+
	t^{b_i}p_i\sum_{c=0}^{d-j-1}\frac{{\log t \choose
			c}}{(\rho\xi)^c}s'_{i,j+c+1} \, .
\end{equation}
We now compare this to coordinate $m$ of our claim, namely 
\begin{equation}
	\label{eq: target Jy}
	(e t)^{b_i} \xi p_i Q_{i,j}(\log (e t))=
	(e t)^{b_i} \xi p_i
	\sum_{c=0}^{d-j}\frac{{\log (e t) \choose
			c}}{(\rho\xi)^c}s'_{i,j+c} \, .
\end{equation}
We compare the right-hand sides of Equations~\eqref{eq: Jy} and~\eqref{eq: target Jy} by comparing the coefficients of $s'_{i,q}$ for $q\in \set{j,\ldots,d}$ (these being the only ones that appear in the expressions).
For $q=j$ we see that in~\eqref{eq: Jy} the number $s'_{i,j}$ occurs in the first summand only, and its coefficient is thus $\xi\rho t^{b_i}p_i$, while in~\eqref{eq: target Jy} it is $(e t)^{b_i}\xi p_i=e^{b_i} t^{b_i}\xi p_i=\rho t^{b_i}\xi p_i$, since $b_i=\log \rho$. Thus, the coefficients are equal.

For $q>j$, write $q=j+c$ with $c\ge 1$; the coefficient at
$s'_{i,j+c}$ in~\eqref{eq: Jy} is then
$$\xi \rho t^{b_i}p_i \frac{{\log t \choose c}}{(\rho\xi)^c}+ t^{b_i}p_i \frac{{\log t \choose {c-1}}}{(\rho\xi)^{c-1}}
=\frac{t^{b_i}\rho \xi p_i}{(\rho\xi)^c}\left({\log t \choose c}+{\log t \choose {c-1}}\right)
=\frac{t^{b_i}\xi\rho p_i}{\xi^c}{\log t+1 \choose c}$$
where the last equality follows from a continuous version of Pascal's identity.
Finally, by noticing that $\log t+1=\log (et)$, it is easy to see that this is the same coefficient as in~\eqref{eq: target Jy}.

\subsection{Proof of Theorem~\ref{thm: P traj cone is invariant}}
\label{subsec: proof that P traj cone is invariant}

Let $t_0\ge 1$. 
By applying Lemma~\ref{lem: P ray is real} to every $p\in \TT$, we
conclude that $P\cdot \C_{t_0}\subseteq \RR^d$. It is then easy to see that
$P\cdot \C_{t_0}$ is definable in $\theoexp$ (note that the only reason the
set $\C_{t_0}$ might fail to be $\theoexp$-definable is that the underlying domain should be $\RR$ and not $\CC$). 

Next, by Lemma~\ref{lem: traj cones are invariant} we have that $J\cdot \C_{t_0}\subseteq \C_{t_0}$. Applying $P$ from the left, we get $PJ\cdot \C_{t_0}\subseteq P\cdot \C_{t_0}$. Thus, we have $AP\cdot \C_{t_0}=PJP^{-1}P\cdot \C_{t_0}=PJ\cdot \C_{t_0}\subseteq P\cdot \C_{t_0}$.

Finally, observe that $\set{A^n s \mid n\ge \log t_0}\subseteq
P\cdot \C_{t_0}$. Since any finite subset of $\Orb$ can be described in
$\theoreals$, we conclude that the set $\set{A^n s \mid n< \log t_0}\cup P\cdot \C_{t_0}$ is an $\theoexp$-definable invariant for $(A, s)$.

\section{O-Minimal Invariants Must Contain Trajectory Cones}
\label{sec: minimal}
In this section we consider invariants definable in o-minimal extensions of \theoexp. Fix such an extension $\theo$ for the remainder of this section.

\begin{theorem}
	\label{thm:inv contain quasi cone}
	Consider an $\theo$-definable invariant $\Inv$ for the LDS $(A, s)$. Then there exists $t_0 \ge 1$ such that $P\cdot \C_{t_0}\subseteq \Inv$.
\end{theorem}

To prove Theorem~\ref{thm:inv contain quasi cone}, we begin by making
following claims of increasing strength:
\begin{claim}
	\label{clm:ev-in-or-out}
	For every $p \in \TT$
	there exists $t_0 \ge 1$ such that $P\cdot \ray(p,t_0)\subseteq \Inv$ or $P\cdot \ray(p,t_0)\cap \Inv=\emptyset$.
\end{claim}

\begin{claim}
	\label{clm:ev-in}
	For every $p \in \TT$
	there exists $t_0 \ge 1$ such that $P\cdot \ray(p,t_0)\subseteq \Inv$.
\end{claim}

\begin{claim}
	\label{clm:in-uniform}
	There exists $t_0 \ge 1$
	such that for every $p\in \TT$
	we have
	$P\cdot \ray(p,t_0)\subseteq \Inv$. 
\end{claim}

\begin{proof}[Proof of Claim~\ref{clm:ev-in-or-out}]
	Fix $p \in \TT$.  Then the set
\[ \{ t \geq 0: P (t^{b_1}p_1Q_{1,1}(\log t),\ldots,t^{b_k}p_kQ_{k,d_k}(\log t))^\top \in \Inv\}\]
is $\theo$-definable and hence comprises a finite union of intervals.
If this set contains an unbounded interval then there exists $t_0$
such that $P\cdot \ray(p,t_0)\subseteq \Inv$; otherwise there exists
$t_0$ such that $P\cdot \ray(p,t_0)\cap \Inv=\emptyset$.
\end{proof}

Before proceeding to Claim~\ref{clm:ev-in}, we prove an auxiliary lemma, which is an adaptation of a similar result from~\cite{FOOPW17}.
For a set $X$, we write $\overline X$ to refer to the topological closure of $X$.
We use the usual topology on $\RR^n$, $\CC^n$, and the (usual) subspace topology
on their subsets.

\begin{lemma}
	\label{lem:closure}
	Let $S, F \subseteq \TT$ be $\theo$-definable\footnote{Recall that, in order to reason about $\TT\subseteq \CC^{k}$ in $\theo$, we identify $\CC$ with $\RR^2$.} sets such that $\overline{S} = \overline{F} = \TT$.
	Then $F \cap S \ne \emptyset$.
\end{lemma}

\begin{proof}
	We start by stating two properties of the dimension of a definable set in an o-minimal theory $\theo$. 
	First, for any $\theo$-definable set $X\subseteq \RR^n$ we have $\dim(X) = \dim(\overline{X})$ \cite[Chapter 4, Theorem 1.8]{dries_1998}. 
	Secondly, if $X\subseteq Y$ are $\theo$-definable subsets of $\RR^n$ that have the same dimension,
	then $X$ has non-empty interior in $Y$ \cite[Chapter 4, Corollary 1.9]{dries_1998}.
	In the situation at hand, since $\dim(F) = \dim(\overline{F})$, it follows that $F$ has non-empty interior with respect to the subspace topology on $\overline{F}$ = $\overline{S}$. But then $S$ is dense in $\overline{S}$ while $F$ has non-empty interior in $\overline{S}$, and thus $S\cap F\neq \emptyset$.        	
\end{proof}

\begin{proof}[Proof of Claim~\ref{clm:ev-in}]
	We strengthen Claim~\ref{clm:ev-in-or-out}.
	Assume by way of contradiction that there exist $p\in \TT$ and $t_0\in \RR$ such that  $P\cdot \ray(p,t_0)\cap \Inv=\emptyset$, and consider $J^{-1} \cdot \ray(p,t_0)$. Let $q\in \TT$ be $L^{-1}p=(\frac{p_1}{\xi_1},\ldots,\frac{p_k}{\xi_k})$ and let $t_1= \frac{t_0}{e}$. Then $p = L q$ and $t_0 = e t_1$ and, by Lemma~\ref{lem:J on ray}, $J \ray(q, t_1) = \ray(L q, e t_1) = \ray(p, t_0)$.
	Since $J$ is invertible, we conclude that $J^{-1}\ray(p,t_0)=\ray(q,t_1)$.
	
	We now claim that $P\cdot \ray(q,t_1)\cap \Inv =\emptyset$. Recall that $P\cdot \ray(p,t_0)\cap \Inv=\emptyset$. Applying $A^{-1}=P J^{-1} P^{-1}$, we have by the above that $P\cdot \ray(q,t_1)\cap A^{-1} \Inv=\emptyset$.  Since $A \Inv\subseteq \Inv$, then $\Inv\subseteq A^{-1}\Inv$, so we have $P\cdot \ray(q,t_1)\cap \Inv \subseteq P\cdot \ray(q,t_1)\cap A^{-1}\Inv=\emptyset$. 
	
	Clearly $t_1\le t_0$ and $\ray(q,t_0)\subseteq \ray(q,t_1)$, so in particular $P\cdot \ray(q,t_0)\cap \Inv=\emptyset$. Thus, assuming $P\cdot \ray(p, t_0) \cap \Inv = \emptyset$, we have just proved that $P\cdot \ray(L^{-1} p, t_0) \cap \Inv = \emptyset$; repeating this argument, we get that for every $n\in \NN$, the point $s=L^{-n}p$ satisfies $P\cdot \ray(s,t_0)\cap \Inv=\emptyset$.
	
	Let $S=\set{L^{-n}p \mid n\in \NN}$. Then $S$ is dense in $\TT$, since the group of multiplicative relations defined by the eigenvalues of $L^{-1}$ is the same as the one defined by those of $L$.
	Define $S'=\set{s\in \TT \mid P\cdot \ray(s,t_0)\cap \Inv=\emptyset}$. Then $S'$ is $\theo$-definable, and we have $S\subseteq S'\subseteq \TT$. Moreover, $\overline{S}=\TT$, so $\overline{S'}= \TT$. 
	
	We now prove that, in fact, $S' = \TT$.  Assuming (again by
	way of contradiction) that there exists $q\in \TT\setminus
	S'$, then by the definition of $S'$ we have $P\cdot \ray(q,t_0)\cap
	\Inv\neq \emptyset$. It follows that for every $n\in \NN$, the
	point $q'=L^nq$ also satisfies $P\cdot \ray(q',t_0)\cap \Inv\neq
	\emptyset$. Define $F=\set{L^nq \mid n\in \NN}$, then $F$ is
	dense in $\TT$. But then the set $F'=\set{q\in \TT \mid P \cdot \ray(q,t_0)\cap \Inv\neq \emptyset}$ satisfies $F\subseteq
	F'\subseteq \TT$ and $\overline{F'}=\TT$. Now the sets $S'$
	and $F'$ are both definable in \theo, and the topological
	closure of each of them is $\TT$. It follows from
	Lemma~\ref{lem:closure} that $F'\cap S'\neq \emptyset$, which is clearly a contradiction.
	Therefore, there is no $q \in \TT \setminus S'$; that is, $S' = \TT$.
	
	From this, however, it follows that $P\cdot \C_{t_0}\cap \Inv=\emptyset$, which is again a contradiction,
	since $P\cdot \C_{t_0} \cap \Orb \neq \emptyset$ and $\Orb\subseteq \Inv$, so we are done.
\end{proof}

\begin{proof}[Proof of Claim~\ref{clm:in-uniform}]
	Consider the function $f\colon\TT\to \RR$ defined by
	$f(p)=\inf\{t\in \RR \mid P\cdot \ray(p,t)$ $\subseteq \Inv\}$. By Claim~\ref{clm:ev-in} this function is well-defined. Since $P\cdot \ray(p,t)$ is $\theo$-definable, then so is $f$. Moreover, its graph $\Gamma(f)$ has finitely many connected components, and the same dimension as $\TT$. Thus, there exists an open set $K\subseteq \TT$ (in the induced topology on $\TT$) such that $f$ is continuous on $K$. Furthermore, $K$ is homeomorphic to $(0,1)^{m}$ for some $0\le m\le k$, and thus we can find sets $K''\subseteq K'\subseteq K$ such that $K''$ is open, and $K'$ is closed.\footnote{In case $m=0$, the proof actually follows immediately from Claim~\ref{clm:ev-in}, since $\TT$ is finite.} 
	Since $f$ is continuous on $K$, it attains a maximum on
	$K'$. Consider the set $\set{L^n\cdot K'' \mid n\in \NN}$. By
	the density of $\set{L^n \mid n\in\NN}$ in $\TT$, this is an
	open cover of $\TT$, and hence there is a finite subcover $\set{L^{n_1}K'',\ldots,L^{n_m}K''}$. Since $K''\subseteq K'$, it follows that $\set{L^{n_1}K',\ldots,L^{n_m}K'}$ is a finite closed cover of $\TT$.
	
	We now show that, for all $p \in \TT$, we have $f(L p) \le e \cdot f(p)$. Indeed, consider any $p \in \TT$ and $t > 0$ such that $P\cdot \ray(p, t) \sset \Inv$. Applying $A = P J P^{-1}$, we get $P J \cdot \ray(p, t) \sset A \Inv \sset \Inv$. By Lemma~\ref{lem:J on ray}, $J \cdot \ray(p,t)= \ray(L p, e t)$, so we can conclude that $P\cdot \ray(L p, e t) \sset \Inv$. This means that $P\cdot \ray(p, t) \sset \Inv$ implies $P\cdot \ray(L p, e t) \sset \Inv$; therefore, $f(L p) \le e \cdot f(p)$.
	
	Now denote $s_0=\max_{p\in K'} f(p)$. Then for every $1\le i\le m$ we have $\max_{p\in L^{n_i}K'}f(p)\le e^{n_i} s_0$; so $f(p)$ is indeed bounded on $\TT$.
\end{proof}

Finally, we conclude from Claim~\ref{clm:in-uniform} that there exists $t_0\ge 1$ such that $P\cdot \C_{t_0}\subseteq \Inv$. This completes the proof of Theorem~\ref{thm:inv contain quasi cone}.

\section{Deciding the Existence of O-Minimal Invariants}
\label{sec: decidability}
We now turn to the algorithmic aspects of invariants and present our
two main results, Theorems~\ref{thm: theo exp decidable schanuel} and \ref{thm: qsa invariants against sa decidable}.

Let $\theo$ be either $\theoreals$ or $\theoexp$. We consider the
following problem: given an LDS $(A, s)$, with $A\in \QQ^{d\times d}$
and $s\in \QQ^d$, and given an $\theo$-definable halting set
$F \subseteq \RR^d$, we wish to decide whether there exists an
o-minimal invariant $\Inv$ for $(A, s)$ that avoids $F$, and to compute such an invariant if it exists.
We term this question the \emph{O-Minimal Invariant Synthesis Problem for
	$\theo$-Definable Halting Sets}.

The following is a consequence of
Theorems~\ref{thm: P traj cone is invariant} and~\ref{thm:inv contain quasi cone}.

\begin{lemma}
	\label{lem: condition for existence of R inv}
	Let $(A,s)$ and $\theo$ be as above, and let $F$ be
	$\theo$-definable. Then there exists an o-minimal invariant
	$\Inv$ for $(A, s)$ that avoids $F$ iff there is some $t_0 \ge 1$ such
	that $P\cdot \C_{t_0}\cap F=\emptyset$ and such that $A^n s\notin F$ for
	every $0\le n\le \log t_0$.
	\end{lemma}

\begin{proof}
	By Theorem~\ref{thm:inv contain quasi cone}, if an o-minimal
	invariant \Inv for $(A, s)$ exists, then there exists $t_0 \ge 1$
	such that $P\cdot \C_{t_0}\subseteq \Inv$. Moreover,
	$\Inv \cap F = \emptyset$ implies $\Orb \cap F = \emptyset$, so that
	$A^n s \not \in F$ for every~$n\in \NN$, and in particular for
	$0\le n\le \log t_0$.
	
	Conversely, let there be $t_0 \ge 1$ such that 
	$P\cdot \C_{t_0}\cap F=\emptyset$ and, for every
	$0\le n\le \log t_0$, it holds that $A^n s \notin F$. Let 
	$t'_0 \in \mathbb{Q}$ be such that $t'_0 \geq t_0$. By
	Theorem~\ref{thm: P traj cone is invariant}, the set
	$\Inv=P\cdot \C_{t'_0}\cup \set{A^n s \mid 0\le n\le \log t'_0}$ is an
	\theoexp-definable
	invariant that avoids $F$.
\end{proof}

Observe that the formula
$\exists t_0 \ge 1: P\cdot \C_{t_0}\cap F=\emptyset$ is a sentence in
$\theoexp$, and by Lemma~\ref{lem: condition for existence of R inv},
deciding the existence of an invariant amounts to determining the
truth value of this sentence.

\subsection{Decidability for \theoexp-definable halting sets assuming Schanuel's conjecture.}
Applying Theorem~\ref{thm: P traj cone is invariant},
we note that an invariant for $(A, s)$ that avoids $F$---if one exists---%
can always be defined in $\theoexp$.

\begin{theorem}
	\label{thm: theo exp decidable schanuel}
	The O-Minimal Invariant Synthesis Problem for $\theoexp$-Definable
	Halting Sets is decidable, assuming Schanuel's conjecture. 
	Moreover, in positive instances, we can explicitly define such an invariant in \theoexp.
\end{theorem}
\begin{proof}
	Assume Schanuel's conjecture. Then
        by~\cite{MacintyreWilkie1996}, the first-order theory of the
        strucure $\theoexp$ is decidable. Thus we can decide whether
        there exists $t_0\ge 1$ such that $P\cdot \C_{t_0}\cap
        F=\emptyset$. If the sentence is false, then by
        Lemma~\ref{lem: condition for existence of R inv} there is no
        invariant, and we are done.  If the sentence is true, however,
        it still remains to check whether $A^n s \not \in F$ for every
        $0\le n\le \log t_0$. While we can decide whether $A^ns \in F$
        for a fixed $n$, observe that we do not have an \emph{a
          priori} bound on $t_0$. Hence we proceed as follows: For
        every $n\in 1,2,\ldots$, check both whether $A^n s\in F$ and,
        for $t_0=e^n$, whether $PC_{t_0}\cap F =\emptyset$. In case
        $A^n s \in F$, then clearly there is no invariant, since
        $\Orb\cap F\neq \emptyset$, and we are done. On the other
        hand, if $PC_{t_0}\cap F =\emptyset$, then return the
	\theoexp -definable 
	invariant as per
	Lemma~\ref{lem: condition for existence of R inv}.
	
	We claim that the above procedure always halts. Indeed, we know that there exists $t_0$ for which
	$P\cdot \C_{t_0}\cap F =\emptyset$. Thus, either for some
	$n<\log t_0$, it holds that $A^ns\in F$, in which case there
	is no invariant and we halt when we reach $n$, or we proceed until
	we reach $n\ge \log t_0$, in which case we halt and return
	the invariant.
\end{proof}
\begin{remark}
	It is interesting to note that, should Schanuel's conjecture turn out
	to be false, the above procedure could still never return a `wrong'
	invariant. The worst that could happen is that decidability of
	$\theoexp$ fails in that the putative algorithm
	of~\cite{MacintyreWilkie1996} simply never halts, so no verdict is
	ever returned. 
\end{remark}

\subsection{Unconditional decidability for semi-algebraic halting sets.}
\label{subsec: decidability for semialgebraci halting}
In this section we restrict attention to semi-algebraic halting sets. Our main result is as follows.

\begin{theorem}
	\label{thm: sa invariants against sa decidable}
	The O-Minimal Invariant Synthesis Problem for Semi-Algebraic
	Halting Sets is decidable. 
	Moreover, in positive instances, we can explicitly define such an invariant in \theoreals.
\end{theorem}
Theorem~\ref{thm: sa invariants against sa decidable} may come as a double surprise: first, as shown by Lemma~\ref{lem: condition for existence of R inv}, deciding the existence of an o-minimal invariant amounts to determining the truth value of the sentence $\exists t_0\ge 1:P\cdot \C_{t_0}\cap F=\emptyset$. Since $P\cdot \C_{t_0}$ might not be definable in \theoreals, then this sentence is only \theoexp-definable, even when $F$ is \theoreals-definable. Therefore, determining the truth value of this sentence is not immediate. 
Second, even if we do manage to determine the truth value of this sentence, the synthesized invariant as per Theorem~\ref{thm: P traj cone is invariant} (namely $P\cdot \C_{t_0}$ along with a finite ``tail'') is \theoexp definable, but might not be \theoreals definable. Thus, the synthesized invariant in Theorem~\ref{thm: sa invariants against sa decidable} must be more elaborate than $P\cdot \C_{t_0}$.

We prove Theorem~\ref{thm: sa invariants against sa decidable} in parts, first addressing the decidability of the existence of an o-minimal invariant (Theorem~\ref{thm: qsa invariants against sa decidable} below), and then showing how to synthesize, in positive instances, a \theoreals-definable invariant (Section~\ref{subsubsec: sa inv for sa target}).

\begin{theorem}
	\label{thm: qsa invariants against sa decidable}
	The O-Minimal Invariant Synthesis Problem for Semi-Algebraic
	Halting Sets is decidable. 
	Moreover, in positive instances, we can explicitly define such an invariant in \theoexp.
\end{theorem}

By Lemma~\ref{lem: condition for existence of R inv}, in order to prove Theorem~\ref{thm: qsa invariants against sa decidable}, it is enough to decide the truth value of the $\theoexp$-sentence $\exists t_0\ge 1: P\cdot \C_{t_0}\cap F=\emptyset$.  
Indeed, as $A^ns\in \QQ^d$, one can always check unconditionally
whether for a given $n$ the vector $A^n s$ belongs to the
semi-algebraic set $F$. The algorithm is then otherwise the same as
that presented in the proof of Theorem~\ref{thm: theo exp decidable
	schanuel}.
The proof of Theorem~\ref{thm: qsa invariants against sa
	decidable} therefore boils down to the following lemma.

\begin{lemma}
	\label{lem: eventually out of sa set}
	For $F$ a semi-algebraic set, it is decidable whether there exists $t_0 \ge 1$ such that $P\cdot \C_{t_0}\cap F=\emptyset$.
\end{lemma}

Our key tool is the following celebrated result from transcendental
number theory:
\begin{theorem}[Baker's theorem~\cite{baker1993logarithmic}]
	\label{thm:baker}
	Let $\alpha_1,\ldots,\alpha_m\in \CC$ be algebraic numbers
	different from 0 and let $b_1,\ldots,b_m\in \ZZ$ be
	integers. Write $\Lambda=b_1\log \alpha_1+\ldots + b_m\log \alpha_m$. There exists a number $C$ effectively computable from $b_1, \ldots, b_m, \alpha_1, \ldots, \alpha_m$ such that if $\Lambda\neq 0$ then $|\Lambda|>H^{-C}$, where $H$ is the maximum height of $\alpha_1,\ldots,\alpha_m$.
\end{theorem}

Recall that the subgroup $\TT$ of the torus defined by the
multiplicative relations of the eigenvalues of $A$ is a semi-algebraic
set.  Write $\vectau(t) = (t^{b_1} Q_{1,1}(\log  t), \ldots,
t^{b_k} Q_{k,d_k}(\log t))$, and consider the set
$$U=\set{\left(
	z_1,
	\ldots,
	z_d
	\right)^\tp\in \CC^d
	\mid
	\forall (p_1,\ldots,p_d)\in \TT,\ 		 			
	P\left(
	z_1 p_1,
	\ldots,
	z_d p_d
	\right)^\tp\in \RR^d \setminus F} \, .$$
It is enough to decide whether there exists $t_0 \ge 1$ such that for all $t\ge t_0$, $\vectau(t) \in U$.

Observe that $U$ is a semi-algebraic set (see Remark~\ref{rmk:complex definable} on
Page~\pageref{rmk:complex definable}). 

Using cell decomposition, describe $U$ as a finite union of connected
components, each of which is given by a conjunction of the form
$\bigwedge_{l=1}^{m} R_l(u_1,\ldots,u_d, v_1,\ldots,v_d)\sim_l
0$. Here, for every $1\le l\le m$, $\mbox{$\sim_l$} \in \set{>,=}$ and $R_l$ is a polynomial with integer coefficients in variables $u_1,\ldots,u_d, v_1, \ldots, v_d$; for each $i$, the variables $u_i$ and $v_i$ represent $\Re z_i$ and $\Im z_i$, the real and imaginary parts of $z_i$, respectively.

We claim that we can restrict our attention to a single connected component. Indeed, first note that by substituting $\vectau(t)$ for $(z_1, \ldots, z_d)$ in the conjunction $\bigwedge_{l=1}^{m} R_l \sim_l 0$, we get a constraint on $t$ expressible in \theoexp. By o-minimality of \theoexp, the set of all $t \in \RR$ satisfying this conjunction is a finite union of points and (possibly unbounded) intervals. Therefore, since the number of connected components is finite, the following two statements are equivalent: (i) there exists $t_0 \ge 1$ such that for all $t\ge t_0$ it holds that $\vectau(t) \in U$, and (ii) there exists a single connected component of $U$ for which this holds (perhaps with a larger value of $t_0$).

Thus, we now need to decide whether we can find $t_0 \ge 1$ such that
for every $t\ge t_0$ it holds that $R_l(\vectau(t))\sim_l 0$ for every
$1\le l\le m$. Fix $1\le l\le m$.  Recall that we consider every
vector in $\CC^d$ as a vector in $\RR^{2 d}$; thus, the polynomial
$R_l$ has the form
\begin{equation*}
	\sum_i
	a_i \cdot u_1^{n'_{i,1}} \cdot \ldots \cdot u_d^{n'_{i,d}}
	\cdot v_1^{n''_{i,1}} \cdot \ldots \cdot v_d^{n''_{i,d}} \, ,
\end{equation*}
with $a_i\in \ZZ$ and $n'_{i,s}, n''_{i,s} \in \ZZ_{\ge 0}$.
Therefore, $R_l(\vectau(t))$ is the sum of terms of the form
\begin{align*}
	a_i \cdot t^{(n'_{i,1} + n''_{i,1}) \cdot b_1+\ldots + (n'_{i,d} + n''_{i,d}) \cdot b_k} \cdot
	& (\Re Q_{1,1}(\log t))^{n'_{i,1}} \cdot \ldots \cdot (\Re Q_{k,d_k}(\log t))^{n'_{i,k}} \cdot \\
	& (\Im Q_{1,1}(\log t))^{n''_{i,1}} \cdot \ldots \cdot (\Im Q_{k,d_k}(\log t))^{n''_{i,k}}
\end{align*}
where $Q_{i,j}(\cdot)$, as above, are polynomials from the definition of trajectory cones.
Note that all $Q_{i,j}$ are only evaluated at real points, and hence it is easy
for us to refer to $\Re Q_{i,j}$ and $\Im Q_{i,j}$; these are
polynomials in one real variable with real algebraic coefficients.
We rewrite $R_l(\vectau(t))$ in the form
\begin{equation*}
	\sum_i
	t^{n_{i,1} \cdot b_1+\ldots + n_{i,k} \cdot b_k} \cdot
	f_i(\log t)
\end{equation*}
where each $f_i(\cdot)$ is also a polynomial with real algebraic
coefficients, and $b_1,\ldots,b_k$ are distinct logarithms of the
moduli of the eigenvalues of $A$.
We can compute all these polynomials $f_i$, eliminating from the sum all terms
that have $f_i \equiv 0$.

Observe that $R_l(\vectau(t))$ is a function of a single variable $t >
0$. In order to reason about the sign of this expression as $t\to
\infty$, we need to determine its leading term. To that end, we first
need to decide for every $i\neq j$ whether the two new exponents
${n_{i,1} b_1+\ldots +n_{i,k}b_k}$ and ${n_{j,1} b_1+\ldots
  +n_{j,k}b_k}$ are equal and, if not, which is larger. (If the
exponents are equal, we aggregate the polynomials $f_i$ and $f_j$
accordingly.) By rearranging the terms, it's enough to decide whether
$n_1 b_1+\ldots +n_k b_k>0$ for some $n_1,\ldots,n_k\in \ZZ$. Recall
that $b_j=\log \rho_j$. By Baker's theorem, there exists an
effectively computable $\epsilon>0$ such that either $n_1 b_1+\ldots
+n_k b_k=0$, or $|n_1 b_1+\ldots +n_k b_k|>\epsilon$.

We now proceed by computing an approximation $\Delta$ of $n_1
b_1+\ldots +n_k b_k$ with additive error at most
$\frac{\epsilon}{3}$. This is easily done, as we are dealing with
computable quantities. We then have that $\Delta\in
[-\frac{\epsilon}{3},\frac\epsilon3]$ iff $n_1 b_1+\ldots +n_k b_k=0$,
and otherwise we have $\textrm{sign}(\Delta)=\textrm{sign}(n_1
b_1+\ldots +n_k b_k)$.  Thus we can sort the exponents $n_{i,1} \cdot
b_1+\ldots + n_{i,k} \cdot b_k$ in descending order and, using the
same procedure, compare each of them to $0$.

Now consider the term that has the largest exponent, $m$;
suppose this term is $t^m \cdot f_i(\log t)$.
Then the sign of $R_l(\vectau(t))$ as $t \to \infty$ is determined by the sign
of the leading term of the polynomial $f_i(\cdot)$;
only if the sum is empty can the sign of $R_l(\vectau(t))$ be $0$ for
all sufficiently large $t$.

The argument above shows that
we can compute the leading terms of the expressions $R_l(\vectau(t))$ and decide whether the conjunction $\bigwedge_{l=1}^{m} R_l \sim_l 0$ holds for all $t \ge t_0$ starting from some $t_0$. This completes the proof.

\subsubsection{Existence of Semi-Algebraic Invariants for Semi-Algebraic Halting Sets}
\label{subsubsec: sa inv for sa target}
We now proceed to show that in the case of a semi-algebraic halting set, the existence of an o-minimal invariant implies the existence of a semi-algebraic invariant (note that clearly the other implication is trivial). 

Intuitively, the trajectory cone $\C_{t_0}$ is not already a
semi-algebraic set for two ``reasons'': the $\log t$ factors, and the
possibly-irrational exponents $b_i$. In the following, we
over-approximate these factors by semi-algebraic components. However,
as we will show, the approximation must carefully take into account
the relationships between the exponents.

Let $\epsilon\in (0,1)$, and consider the trajectory cone
$\C_{t_0}$. We define an over-approximating set of $\C_{t_0}$ by
replacing the $\log t$ factors by an interval ranging from some
constant lower bound $\mu$ to $t^\epsilon$. That is, define for
$\mu\in \QQ$ and $t_0\in \RR$ \newcommand{\Capp}{\widetilde{{\C}}}
$$\Capp_{t_0,\epsilon,\mu}=\set{\begin{pmatrix}
	t^{b_1} p_1 Q_{1,1}(s)\\
	\vdots\\
	t^{b_k} p_k Q_{k,d_k}(s)\\
	\end{pmatrix}\ :\ (p_1,\ldots,p_{d})\in \TT,\ t\ge t_0, \mu\le s\le t^{\epsilon}}.$$	

Next, we modify the irrational $b_i$ exponents into rational ones. This is done in two parts. First, we approximate $b_i$ by lower and upper rational bounds, next we enforce additive relationships among the approximations.

Consider vectors $\lowervec=(\ell_1,\ldots,\ell_k),\uppervec=(u_1,\ldots,u_k)\in \QQ^k$ with $\lowervec\le \mathbf{b}\le \uppervec$ and define $$\boxvec(\lowervec,\uppervec)=\set{\mathbf{c}\in \RR^k: \lowervec\le \mathbf{c}\le \uppervec}.$$ 
Second, let $\mathbf{b}=(b_1,\ldots,b_k)$, and define 
\[
\SS=\set{\mathbf{c}\in \RR^k: \forall \mathbf{z}\in \ZZ^k,\ \mathbf{b}\cdot \mathbf{z}=0 \text{ implies }\mathbf{c}\cdot \mathbf{z}=0}
\]
to be the set of vectors that maintain the integer additive relationships among the $b_i$. 

We are now ready to define the \emph{fat trajectory cone}, which is a semi-algebraic set that approximates $\Capp_{t_0,\epsilon,\mu}$. Given $t_0,\epsilon,\mu,\lowervec,$ and $\uppervec$, we define
\[\F_{t_0,\epsilon,\mu}^{\lowervec,\uppervec}=\set{\begin{pmatrix}
	t^{c_1} p_1 Q_{1,1}(s)\\
	\vdots\\
	t^{c_k} p_k Q_{k,d_k}(s)\\
	\end{pmatrix}\ :\ (p_1,\ldots,p_{d})\in \TT,\ t\ge t_0, \mu\le s\le t^{\epsilon}, \mathbf{c}\in \SS\cap \boxvec(\lowervec,\uppervec)}.\]
That is, we replace the exponents vector $\mathbf{b}$ with a vector $\mathbf{c}$ that is close enough to $\mathbf{b}$, and maintains its additive relations.

It is not immediate from the definition of
$\F_{t_0,\epsilon,\mu}^{\lowervec,\uppervec}$ that it is indeed a
semi-algebraic set, nor that we can compute a representation of
it. Indeed, we cannot quantify the exponents $\mathbf{c}$ in the
first-order theory of the reals, and it is not clear that the set
$\SS$ can be finitely represented. We start by addressing these
issues.
\begin{lemma}
	\label{lem: fat cone is semialgebraic}
	$\F_{t_0,\epsilon,\mu}^{\lowervec,\uppervec}$ is definable in $\theoreals$, and we can compute a representation of it.
\end{lemma}
\begin{proof}
	Recall that $\mathbf{b}=(b_1,\ldots,b_k)$ where $b_i=\log
        \rho_i$ for every $i$.  Consider the abelian group
        $L=\set{\mathbf{z}\in \ZZ^k: \varrho^\mathbf{z}=1}$ where
        $\varrho^\mathbf{z}=\rho_1^{z_1}\cdots
        \rho_k^{z_k}$. By~\cite{Mas88} we can compute a finite basis
        $\set{\mathbf{z}^1,\ldots,\mathbf{z}^m}$ for $L$.
	
	Note that for every $\mathbf{z}\in \ZZ^k$ we have that $\mathbf{b}\cdot \mathbf{z}=0$ iff $\varrho^\mathbf{z}=1$. Thus, we can write
	\[
	\SS=\set{\mathbf{c}\in \RR^k: \forall \mathbf{z}\in \ZZ^k, \mathbf{z}\in L \text{ implies } \mathbf{c}\cdot \mathbf{z}=0}
	=\set{\mathbf{c}\in \RR^k: \bigwedge_{i=1}^m\mathbf{c}\cdot \mathbf{z}^i=0}.
	\]
	Let $\SS'=\set{\mathbf{r}\in \RR^k: \bigwedge_{i=1}^m\mathbf{r}^{\mathbf{z}^i}=1}$ 
	where, as before, $\mathbf{r}^\mathbf{z}=r_1^{z_1}\cdots r_k^{z_k}$.
	Consider some $t>1$. For every $\mathbf{c}\in \RR^k$, denote $t^\mathbf{c}=(t^{c_1},\ldots,t^{c_k})$, then clearly $\mathbf{c}\in \SS$ iff $t^{\mathbf{c}}\in \SS'$.  
	In addition, $\mathbf{c}\in \boxvec(\lowervec,\uppervec)$ iff $t^\mathbf{c}\in \boxvec(t^{\lowervec},t^{\uppervec})$.
	We conclude that we can write
	\[
	\F_{t_0,\epsilon,\mu}^{\lowervec,\uppervec}=\set{\begin{pmatrix}
		r_1 p_1 Q_{1,1}(s)\\
		\vdots\\
		r_k p_k Q_{k,d_k}(s)\\
		\end{pmatrix}\ :\ (p_1,\ldots,p_{d})\in \TT,\ t\ge t_0, \mu\le s\le t^{\epsilon}, \mathbf{r}\in \SS'\cap \boxvec(t^{\lowervec},t^{\uppervec}) }.
	\]
	Since $\TT$, $\SS'$ and $\boxvec(t^{\lowervec},t^{\uppervec})$ are all semi-algebraic sets (the latter due to $\lowervec$ and $\uppervec$ being rational vectors), then so is $\F_{t_0,\epsilon,\mu}^{\lowervec,\uppervec}$.
	
\end{proof}

The following lemma is the main result of this section, and states that if $P\cdot \C_{t_0}$ is disjoint from a semi-algebraic halting set, then we can approximate it by some appropriate fat cone.
\begin{lemma}
	\label{lem: fat cone sa inv}
	Let $Y$ be a semi-algebraic set such that $P\cdot \C_{t_0}\cap Y= \emptyset$ for some $t_0\in \RR$, then  there exist $\mu,t_1\in \RR$, $\epsilon>0$, and $\lowervec,\uppervec\in \QQ^d$ such that
	$P\cdot \F_{t_1,\epsilon,\mu}^{\lowervec,\uppervec} \cap Y=\emptyset$ and $A\cdot P\cdot \F_{t_1,\epsilon,\mu}^{\lowervec,\uppervec}\subseteq P\cdot \F_{t_1,\epsilon,\mu}^{\lowervec,\uppervec}$.
\end{lemma}

The rest of the section is devoted to the proof of Lemma~\ref{lem: fat cone sa inv}.
We start by showing that the fat cone is invariant under $J$.
\begin{lemma}
	\label{lem: fat cone invariant}
	For every $\epsilon$, there exists $t_0$ such that for every $t_1\ge t_0$ and for every $\mu,\lowervec,\uppervec$ we have that 
	$J\cdot \F_{t_1,\epsilon,\mu}^{\lowervec,\uppervec}\subseteq \F_{t_1,\epsilon,\mu}^{\lowervec,\uppervec}$.
\end{lemma} 
\begin{proof}
	Following the proof of Lemma~\ref{lem:J on ray}, we see that for $y=\begin{pmatrix}
	t^{c_1} p_1 Q_{1,1}(s)\\
	\vdots\\
	t^{c_k} p_k Q_{k,d_k}(s)\\
	\end{pmatrix}$ we have that $Jy= \begin{pmatrix}
	e^{b_1} t^{c_1} \xi_1 p_1 Q_{1,1}(s+1)\\
	\vdots\\
	e^{b_k} t^{c_k} \xi_k p_k Q_{k,d_k}(s+1)\\
	\end{pmatrix}$.
	Thus, it suffices to show that for large enough $t_0$, for every $t\ge t_0$ the following hold:
	\begin{enumerate}
		\item If $\mu\le s\le t^{\epsilon}$ then $\mu\le s+1\le (e t)^\epsilon$.
		\item $(e^{b_1}t^{c_1},\ldots,e^{b_k}t^{c_k})=((e t)^{c'_1},\ldots,(e t)^{c'_k})$ for some $\mathbf{c'}\in \SS\cap \boxvec(\lowervec,\uppervec)$.
	\end{enumerate} 
	For item 1, observe that since $\mu\le s$, we clearly have
        $\mu\le s+1$. For the second inequality, since $\epsilon>0$,
        we have that $e^{\epsilon}-1>0$. Let $t_0$ be large
        enough such that $(e^\epsilon-1)t^\epsilon\ge 1$ for
        every $t\ge t_0$, then for every $t\ge t_0$ we have that
        $t^\epsilon+1 \le (e t)^\epsilon$. It follows that
        $s+1\le t^\epsilon+1\le (e t)^\epsilon$, as desired.
	
	For item 2, define $\mathbf{c'}=(c'_1,\ldots,c'_k)$ by setting $c'_i=\frac{b_i+\log (t)c_i}{1+\log (t)}$. It is easy to check that $e^{b_i}t^{c_i})=(e t)^{c'_i}$ for every $1\le i\le k$. Furthermore, since $\SS\cap \boxvec(\lowervec,\uppervec)$ is a convex set, and $\mathbf{c'}$ is a convex combination of $\mathbf{b}$ and $\mathbf{c}$, it follows that $\mathbf{c'}\in \SS\cap \boxvec(\lowervec,\uppervec)$, and we are done.
\end{proof}

It remains to prove that we can choose a fat cone that is disjoint from the target set.
\begin{lemma}
	\label{lem: fat cone disjoint}
	Let $Y$ be a semi-algebraic set, and let $t_0\in \RR$ be such that $P \C_{t_0}\cap Y=\emptyset$, then there exist $\epsilon>0$, $\mu\in \RR, t_1\in \RR,$ and $\lowervec,\uppervec\in \QQ^k$ such that $P\cdot  \F_{t_1,\epsilon,\mu}^{\lowervec,\uppervec}\cap Y=\emptyset$.
\end{lemma}
\begin{proof}
	Working along the lines of the proof of Lemma~\ref{lem: eventually out of sa set}, we define \[U=\set{\left(
		z_1,
		\ldots,
		z_d
		\right)^\tp\in \CC^d
		:
		\forall (p_1,\ldots,p_d)\in \TT,\ 		 			
		P\left(
		z_1 p_1,
		\ldots,
		z_d p_d
		\right)^\tp\in \RR^d \setminus Y},\]
	and let $\phi$ be a quantifier-free formula defining $U$. As in the proof of Lemma~\ref{lem: eventually out of sa set}, let $\vectau(t)=(t^{b_1}Q_{1,1}(\log t),\ldots,t^{b_k}Q_{k,d_k}(\log t))$, then for the purpose of evaluating $\phi(\vectau(t))$, we can assume $\phi$ is of the form  $\bigwedge_{l=1}^{m} R_l(u_1,\ldots,u_d, v_1,\ldots,v_d)\sim_l 0$. Since $P\cdot \C_{t_0}\cap Y=\emptyset$, then $R_l(\vectau(t))\sim_l 0$ for every $l$ and every $t\ge t_0$. Writing 
	\begin{equation}
		\label{eq: Rl semialgebraic inv}
		R_l(\vectau(t))=\sum_i
		t^{n_{i,1} \cdot b_1+\ldots + n_{i,k} \cdot b_k} \cdot
		f_i(\log t),
	\end{equation}
	we show, first, how to replace the exponents vector $\mathbf{b}$ by any exponents vector in $\SS\cap \boxvec(\lowervec,\uppervec)$ for appropriate $\lowervec,\uppervec$, and second, how to replace $\log t$ by $s$ for $\mu\le s\le t^{\epsilon}$ for some appropriate $\mu$ and $\epsilon$, while maintaining the inequality or equality prescribed by $\sim_l$.

	Denote by $N$ the set of vectors $\mathbf{n_i}=(n_{i,1},\ldots,n_{i,k})$ of exponents in~\eqref{eq: Rl semialgebraic inv}. 
	As in the proof of Lemma~\ref{lem: eventually out of sa set}, we can compute $\delta>0$ such that for every $\mathbf{n},\mathbf{n'}\in N$, if $\mathbf{b}\cdot (\mathbf{n}-\mathbf{n'})\neq 0$ then $|\mathbf{b}\cdot (\mathbf{n}-\mathbf{n'})|>\delta$.	
	Let $M=\max_{\mathbf{n,n'}\in N}\norm{\mathbf{n}-\mathbf{n'}}$ (where $\norm{\cdot}$ is the Euclidean norm in $\RR^k$).
	\begin{claim}
		\label{prop: approximate exponent}
		Let $\mathbf{c}\in \RR^d$ be such that $\norm{\mathbf{b}-\mathbf{c}}\le \frac{\delta}{2M}$, then, for all $\mathbf{n},\mathbf{n'}\in N$, if $\mathbf{b}\cdot (\mathbf{n}-\mathbf{n'})> 0$ then $\mathbf{c}\cdot (\mathbf{n}-\mathbf{n'})>\frac{\delta}{2}$.
	\end{claim}
	\begin{proof}[Proof of Claim~\ref{prop: approximate exponent}]
		Suppose that $\mathbf{b}\cdot (\mathbf{n}-\mathbf{n'})> 0$, then by the above we have $\mathbf{b}\cdot (\mathbf{n}-\mathbf{n'})> \delta$, and hence
		\begin{align*}
			\mathbf{c}\cdot (\mathbf{n}-\mathbf{n'})= \mathbf{b}\cdot (\mathbf{n}-\mathbf{n'})+(\mathbf{c}-\mathbf{b})\cdot (\mathbf{n}-\mathbf{n'})
			\ge  \delta-\norm{\mathbf{c}-\mathbf{b}}\cdot\norm{\mathbf{n}-\mathbf{n'}}
			\ge  \delta-\frac{\delta}{2M}M=\frac{\delta}{2}.
		\end{align*}
	\end{proof}
	We can now choose $\lowervec$ and $\uppervec$ such that 
	$u_i-\ell_i\le \frac{\delta}{2M\sqrt{k}}$
	so that for all $\mathbf{c}\in \boxvec(\lowervec,\uppervec)$ we have 
	\[
	\norm{\mathbf{b}-\mathbf{c}}\le \sqrt{\sum_{i=1}^k (u_i-\ell_i)^2}\le \sqrt{\frac{\delta^2}{(2M)^2}}=\frac{\delta}{2M}.
	\]
	It follows from Claim~\ref{prop: approximate exponent} and from the definition of $\SS$ that, intuitively, every $\mathbf{c}\in \boxvec(\lowervec,\uppervec)$ maintains the order of magnitude of the monomials $t^{n_{i,1} \cdot b_1+\ldots + n_{i,k} \cdot b_k}$ in $R_l(\vectau(t))$. 
	
	More precisely, let $\vectau'(t)=(t^{c_1}Q_{1,1}(\log t),\ldots,t^{c_k}Q_{k,d_k}(\log t))$ for some $\mathbf{c}\in \boxvec(\lowervec,\uppervec)$, Then 
	the exponent of the ratio of every two monomials in $R_l(\vectau'(t))$ has the same (constant) sign as the corresponding exponent in $R_l(\vectau(t))$. Moreover, the exponents of distinct monomials in $R_l(\vectau(t))$ differ by at least $\frac{\delta}{2}$ in $R_l(\vectau'(t))$.
	
	We are now ready to handle the $\log t$ factor. First, assume
        $t_0$ is large enough that $f_i(\log t)$ has constant sign for
        every $t\ge t_0$ (otherwise increase $t_0$ accordingly). We
        can now let $\mu$ be large enough that for every $s\ge \mu$,
        the sign of $f_i(\log t)$ coincides with the sign of
        $f_i(s)$ for every $t\ge t_0$. It remains to give an upper
        bound on $s$ of the form $t^\epsilon$ such that plugging
        $f_i(s)$ instead of $f_i(\log t)$ does not change the
        ordering of the terms (by their magnitude) in
        $R_l(\vectau'(t))$.  Let $D$ be the maximum degree of all
        polynomials $f_i$ in~\eqref{eq: Rl semialgebraic inv}, and
        define $\epsilon=\frac{\delta}{3D}$ (in fact, any
        $\epsilon<\frac{\delta}{2D}$ would suffice), then we have
        that, for $t\ge t_0$, $f_i(s)$ has the same sign as
        $f_i(\log t)$ for every $\mu\le s\le t^{\epsilon}$ (by
        our choice of $\mu$), and 
        guarantees that plugging $t^{\epsilon}$ instead of $t$ does
        not change the ordering of the terms (by their magnitude) in
        $R_l$.  since the exponents of the monomials in
        $R_l(\vectau'(t))$ differ by at least $\frac{\delta}{2}$, it
        follows that their order is maintained when replacing
        $\log t$ by $s$.
	
	Let $\vectau''(t)=(t^{c_1}Q_{1,1}(s),\ldots,t^{c_k}Q_{k,d_k}(s))$ for some $\mathbf{c}\in \boxvec(\lowervec,\uppervec)$ and $\mu\le s\le t^\epsilon$, then by our choice of $\epsilon$, the dominant term in $R_l(\vectau''(t))$ is the same as that in $R_l(\vectau(t))$. Therefore, for large enough $t$, the signs of $R_l(\vectau''(t))$ and $R_l(\vectau(t))$ are the same.
	
	By repeating this argument for each $R_l$, we can compute $t_1\in \RR$, $\epsilon>0$, $\mu\in \RR$, and $\lowervec,\uppervec\in \QQ^k$ such that $P\cdot  \F_{t_1,\epsilon,\mu}^{\lowervec,\uppervec}\cap Y=\emptyset$, and we are done.
\end{proof}
We are now ready to complete the proof of Theorem~\ref{thm: sa invariants against sa decidable}.
\begin{proof}[Proof of Theorem~\ref{thm: sa invariants against sa decidable}]
	By Theorem~\ref{thm: qsa invariants against sa decidable}, we can decide whether there exists an o-minimal invariant for the LDS $(A,s)$ that avoids a semialgebraic target $Y$. Moreover, in positive instances we can synthesize such an invariant of the form $\Inv=P\cdot \C_{t_0}\cup \set{A^n s \mid 0\le n\le \log t_0}$ for some $t_0\ge 1$. 
	
	In particular, $P\cdot \C_{t_0}\cap Y = \emptyset$, so by Lemma~\ref{lem: fat cone sa inv}, there exist $\mu, t_1\in \RR$, $\epsilon>0$, and $\lowervec,\uppervec\in \QQ^d$ such that $P\cdot \F_{t_1,\epsilon,\mu}^{\lowervec,\uppervec} \cap Y=\emptyset$ and $P\cdot \F_{t_1,\epsilon,\mu}^{\lowervec,\uppervec}$ is invariant under $A$.  Furthermore, by Lemma~\ref{lem: fat cone is semialgebraic}, $\F_{t_1,\epsilon,\mu}^{\lowervec,\uppervec}$ is semialgebraic.
	
	Naively, one might think that $\F_{t_1,\epsilon,\mu}^{\lowervec,\uppervec}$ with the addition of the ``finite tail'' $\set{A^n s \mid 0\le n\le \log t_1}$ would make up a semialgebraic invariant for $(A,s)$. This, however, may not be the case. Recall that 
	\[
	\F_{t_1,\epsilon,\mu}^{\lowervec,\uppervec}=\set{\begin{pmatrix}
		r_1 p_1 Q_{1,1}(s)\\
		\vdots\\
		r_k p_k Q_{k,d_k}(s)\\
		\end{pmatrix}\ :\ (p_1,\ldots,p_{d})\in \TT,\ t\ge t_1, \mu\le s\le t^{\epsilon}, \mathbf{r}\in \SS'\cap \boxvec(t^{\lowervec},t^{\uppervec}) }.
	\]
	and observe that if $\mu$ is large and $\epsilon$ is small, then for small enough $t\ge t_1$, it may be the case that $t^\epsilon<\mu$, so there does not exist an appropriate $s$. Thus, we may need to extend the finite tail to capture this. 
	
	Let $t_2\ge t_1$ such that for every $n\ge \log t_2$ we have that $A^ns\in \F_{t_1,\epsilon,\mu}^{\lowervec,\uppervec}$ (note that we can compute $t_2$ by iterating over the orbit, until the first point in $\F_{t_1,\epsilon,\mu}^{\lowervec,\uppervec}$ is found), then $\F_{t_1,\epsilon,\mu}^{\lowervec,\uppervec}\cup \set{A^n s \mid 0\le n\le \log t_2}$ is a semialgebraic invariant for $(A,s)$ that avoids $Y$, which concludes the proof.
\end{proof}

\begin{acks}	
  Jo\"el Ouaknine was supported by ERC grant AVS-ISS (648701) and by DFG grant 389792660 as part of
TRR 248 (see https://perspicuous-computing.science). James Worrell was supported by EPSRC Fellowship EP/N008197/1. Shaull Almagor has received funding from the European Union's Horizon 2020 research and innovation programme under the Marie Skłodowska-Curie grant agreement No 837327.
\end{acks}

\bibliographystyle{ACM-Reference-Format}
\bibliography{Main}


\begin{thebibliography}{37}


\ifx \showCODEN    \undefined \def \showCODEN     #1{\unskip}     \fi
\ifx \showDOI      \undefined \def \showDOI       #1{#1}\fi
\ifx \showISBNx    \undefined \def \showISBNx     #1{\unskip}     \fi
\ifx \showISBNxiii \undefined \def \showISBNxiii  #1{\unskip}     \fi
\ifx \showISSN     \undefined \def \showISSN      #1{\unskip}     \fi
\ifx \showLCCN     \undefined \def \showLCCN      #1{\unskip}     \fi
\ifx \shownote     \undefined \def \shownote      #1{#1}          \fi
\ifx \showarticletitle \undefined \def \showarticletitle #1{#1}   \fi
\ifx \showURL      \undefined \def \showURL       {\relax}        \fi
\providecommand\bibfield[2]{#2}
\providecommand\bibinfo[2]{#2}
\providecommand\natexlab[1]{#1}
\providecommand\showeprint[2][]{arXiv:#2}

\bibitem[\protect\citeauthoryear{Baker and W{\"u}stholz}{Baker and
  W{\"u}stholz}{1993}]%
        {baker1993logarithmic}
\bibfield{author}{\bibinfo{person}{Alan Baker} {and} \bibinfo{person}{Gisbert
  W{\"u}stholz}.} \bibinfo{year}{1993}\natexlab{}.
\newblock \showarticletitle{Logarithmic forms and group varieties}.
\newblock \bibinfo{journal}{\emph{J. reine angew. Math}} \bibinfo{volume}{442},
  \bibinfo{number}{19-62} (\bibinfo{year}{1993}), \bibinfo{pages}{3}.
\newblock


\bibitem[\protect\citeauthoryear{Ben{-}Amram and Genaim}{Ben{-}Amram and
  Genaim}{2014}]%
        {BG14}
\bibfield{author}{\bibinfo{person}{Amir~M. Ben{-}Amram} {and}
  \bibinfo{person}{Samir Genaim}.} \bibinfo{year}{2014}\natexlab{}.
\newblock \showarticletitle{Ranking Functions for Linear-Constraint Loops}.
\newblock \bibinfo{journal}{\emph{J. {ACM}}} \bibinfo{volume}{61},
  \bibinfo{number}{4} (\bibinfo{year}{2014}), \bibinfo{pages}{26:1--26:55}.
\newblock


\bibitem[\protect\citeauthoryear{Ben{-}Amram, Genaim, and Masud}{Ben{-}Amram
  et~al\mbox{.}}{2012}]%
        {BGM12}
\bibfield{author}{\bibinfo{person}{Amir~M. Ben{-}Amram}, \bibinfo{person}{Samir
  Genaim}, {and} \bibinfo{person}{Abu~Naser Masud}.}
  \bibinfo{year}{2012}\natexlab{}.
\newblock \showarticletitle{On the Termination of Integer Loops}.
\newblock \bibinfo{journal}{\emph{{ACM} Trans. Program. Lang. Syst.}}
  \bibinfo{volume}{34}, \bibinfo{number}{4} (\bibinfo{year}{2012}),
  \bibinfo{pages}{16:1--16:24}.
\newblock


\bibitem[\protect\citeauthoryear{Braverman}{Braverman}{2006}]%
        {Bra06}
\bibfield{author}{\bibinfo{person}{Mark Braverman}.}
  \bibinfo{year}{2006}\natexlab{}.
\newblock \showarticletitle{Termination of Integer Linear Programs}. In
  \bibinfo{booktitle}{\emph{Computer Aided Verification, 18th International
  Conference, {CAV} 2006, Seattle, WA, USA, August 17-20, 2006, Proceedings}}.
  \bibinfo{pages}{372--385}.
\newblock


\bibitem[\protect\citeauthoryear{Cassels}{Cassels}{1965}]%
        {cassels1965introduction}
\bibfield{author}{\bibinfo{person}{John~W.S. Cassels}.}
  \bibinfo{year}{1965}\natexlab{}.
\newblock \bibinfo{booktitle}{\emph{An Introduction to Diophantine
  Approximation}}.
\newblock \bibinfo{publisher}{Cambridge University Press}.
\newblock


\bibitem[\protect\citeauthoryear{Chonev, Ouaknine, and Worrell}{Chonev
  et~al\mbox{.}}{2013}]%
        {COW13}
\bibfield{author}{\bibinfo{person}{Ventsislav Chonev},
  \bibinfo{person}{Jo{\"{e}}l Ouaknine}, {and} \bibinfo{person}{James
  Worrell}.} \bibinfo{year}{2013}\natexlab{}.
\newblock \showarticletitle{The {O}rbit {P}roblem in higher dimensions}. In
  \bibinfo{booktitle}{\emph{Symposium on Theory of Computing Conference,
  STOC'13, Palo Alto, CA, USA, June 1-4, 2013}}. \bibinfo{pages}{941--950}.
\newblock


\bibitem[\protect\citeauthoryear{Chonev, Ouaknine, and Worrell}{Chonev
  et~al\mbox{.}}{2015}]%
        {COW15}
\bibfield{author}{\bibinfo{person}{Ventsislav Chonev},
  \bibinfo{person}{Jo{\"{e}}l Ouaknine}, {and} \bibinfo{person}{James
  Worrell}.} \bibinfo{year}{2015}\natexlab{}.
\newblock \showarticletitle{The Polyhedron-Hitting Problem}. In
  \bibinfo{booktitle}{\emph{Proceedings of the Twenty-Sixth Annual {ACM-SIAM}
  Symposium on Discrete Algorithms, {SODA} 2015, San Diego, CA, USA, January
  4-6, 2015}}. \bibinfo{pages}{940--956}.
\newblock


\bibitem[\protect\citeauthoryear{Chonev, Ouaknine, and Worrell}{Chonev
  et~al\mbox{.}}{2016}]%
        {COW16}
\bibfield{author}{\bibinfo{person}{Ventsislav Chonev},
  \bibinfo{person}{Jo{\"{e}}l Ouaknine}, {and} \bibinfo{person}{James
  Worrell}.} \bibinfo{year}{2016}\natexlab{}.
\newblock \showarticletitle{On the Complexity of the Orbit Problem}.
\newblock \bibinfo{journal}{\emph{J. {ACM}}} \bibinfo{volume}{63},
  \bibinfo{number}{3} (\bibinfo{year}{2016}), \bibinfo{pages}{23:1--23:18}.
\newblock


\bibitem[\protect\citeauthoryear{Col{\'{o}}n, Sankaranarayanan, and
  Sipma}{Col{\'{o}}n et~al\mbox{.}}{2003}]%
        {CSS03}
\bibfield{author}{\bibinfo{person}{Michael Col{\'{o}}n},
  \bibinfo{person}{Sriram Sankaranarayanan}, {and} \bibinfo{person}{Henny
  Sipma}.} \bibinfo{year}{2003}\natexlab{}.
\newblock \showarticletitle{Linear Invariant Generation Using Non-linear
  Constraint Solving}. In \bibinfo{booktitle}{\emph{Computer Aided
  Verification, 15th International Conference, {CAV} 2003, Boulder, CO, USA,
  July 8-12, 2003, Proceedings}}. \bibinfo{pages}{420--432}.
\newblock


\bibitem[\protect\citeauthoryear{Cousot}{Cousot}{2005}]%
        {Cou05}
\bibfield{author}{\bibinfo{person}{Patrick Cousot}.}
  \bibinfo{year}{2005}\natexlab{}.
\newblock \showarticletitle{Proving Program Invariance and Termination by
  Parametric Abstraction, Lagrangian Relaxation and Semidefinite Programming}.
  In \bibinfo{booktitle}{\emph{Verification, Model Checking, and Abstract
  Interpretation, 6th International Conference, {VMCAI} 2005, Paris, France,
  January 17-19, 2005, Proceedings}}. \bibinfo{pages}{1--24}.
\newblock


\bibitem[\protect\citeauthoryear{Cousot and Halbwachs}{Cousot and
  Halbwachs}{1978}]%
        {CH78}
\bibfield{author}{\bibinfo{person}{Patrick Cousot} {and}
  \bibinfo{person}{Nicolas Halbwachs}.} \bibinfo{year}{1978}\natexlab{}.
\newblock \showarticletitle{Automatic Discovery of Linear Restraints Among
  Variables of a Program}. In \bibinfo{booktitle}{\emph{Conference Record of
  the Fifth Annual {ACM} Symposium on Principles of Programming Languages,
  Tucson, Arizona, USA, January 1978}}. \bibinfo{pages}{84--96}.
\newblock


\bibitem[\protect\citeauthoryear{Dries}{Dries}{1998}]%
        {dries_1998}
\bibfield{author}{\bibinfo{person}{L.~P. D. van~den Dries}.}
  \bibinfo{year}{1998}\natexlab{}.
\newblock \bibinfo{booktitle}{\emph{Tame Topology and O-minimal Structures}}.
\newblock \bibinfo{publisher}{Cambridge University Press}.
\newblock


\bibitem[\protect\citeauthoryear{Fijalkow, Lefaucheux, Ohlmann, Ouaknine,
  Pouly, and Worrell}{Fijalkow et~al\mbox{.}}{2019a}]%
        {FLOOPW19}
\bibfield{author}{\bibinfo{person}{Nathana{\"{e}}l Fijalkow},
  \bibinfo{person}{Engel Lefaucheux}, \bibinfo{person}{Pierre Ohlmann},
  \bibinfo{person}{Jo{\"{e}}l Ouaknine}, \bibinfo{person}{Amaury Pouly}, {and}
  \bibinfo{person}{James Worrell}.} \bibinfo{year}{2019}\natexlab{a}.
\newblock \showarticletitle{On the Monniaux Problem in Abstract
  Interpretation}. In \bibinfo{booktitle}{\emph{Static Analysis, 26th
  International Symposium, {SAS} 2019, Porto, Portugal, October 8-11, 2019,
  Proceedings}}.
\newblock


\bibitem[\protect\citeauthoryear{Fijalkow, Ohlmann, Ouaknine, Pouly, and
  Worrell}{Fijalkow et~al\mbox{.}}{2017}]%
        {FOOPW17}
\bibfield{author}{\bibinfo{person}{Nathana{\"{e}}l Fijalkow},
  \bibinfo{person}{Pierre Ohlmann}, \bibinfo{person}{Jo{\"{e}}l Ouaknine},
  \bibinfo{person}{Amaury Pouly}, {and} \bibinfo{person}{James Worrell}.}
  \bibinfo{year}{2017}\natexlab{}.
\newblock \showarticletitle{Semialgebraic Invariant Synthesis for the
  Kannan-Lipton Orbit Problem}. In \bibinfo{booktitle}{\emph{34th Symposium on
  Theoretical Aspects of Computer Science, {STACS} 2017, March 8-11, 2017,
  Hannover, Germany}}. \bibinfo{pages}{29:1--29:13}.
\newblock


\bibitem[\protect\citeauthoryear{Fijalkow, Ohlmann, Ouaknine, Pouly, and
  Worrell}{Fijalkow et~al\mbox{.}}{2019b}]%
        {FOOPW19}
\bibfield{author}{\bibinfo{person}{Nathana{\"{e}}l Fijalkow},
  \bibinfo{person}{Pierre Ohlmann}, \bibinfo{person}{Jo{\"{e}}l Ouaknine},
  \bibinfo{person}{Amaury Pouly}, {and} \bibinfo{person}{James Worrell}.}
  \bibinfo{year}{2019}\natexlab{b}.
\newblock \showarticletitle{Complete Semialgebraic Invariant Synthesis for the
  Kannan-Lipton Orbit Problem}.
\newblock \bibinfo{journal}{\emph{Theory Comput. Syst.}} \bibinfo{volume}{63},
  \bibinfo{number}{5} (\bibinfo{year}{2019}), \bibinfo{pages}{1027--1048}.
\newblock


\bibitem[\protect\citeauthoryear{Gupta, Henzinger, Majumdar, Rybalchenko, and
  Xu}{Gupta et~al\mbox{.}}{2008}]%
        {GHMRX08}
\bibfield{author}{\bibinfo{person}{Ashutosh Gupta}, \bibinfo{person}{Thomas~A.
  Henzinger}, \bibinfo{person}{Rupak Majumdar}, \bibinfo{person}{Andrey
  Rybalchenko}, {and} \bibinfo{person}{Ru{-}Gang Xu}.}
  \bibinfo{year}{2008}\natexlab{}.
\newblock \showarticletitle{Proving non-termination}. In
  \bibinfo{booktitle}{\emph{Proceedings of the 35th {ACM} {SIGPLAN-SIGACT}
  Symposium on Principles of Programming Languages, {POPL} 2008, San Francisco,
  California, USA, January 7-12, 2008}}. \bibinfo{pages}{147--158}.
\newblock


\bibitem[\protect\citeauthoryear{Hosseini, Ouaknine, and Worrell}{Hosseini
  et~al\mbox{.}}{2019}]%
        {HOW19}
\bibfield{author}{\bibinfo{person}{Mehran Hosseini},
  \bibinfo{person}{Jo{\"{e}}l Ouaknine}, {and} \bibinfo{person}{James
  Worrell}.} \bibinfo{year}{2019}\natexlab{}.
\newblock \showarticletitle{Termination of Linear Loops over the Integers}. In
  \bibinfo{booktitle}{\emph{46th International Colloquium on Automata,
  Languages, and Programming, {ICALP} 2019, July 9-12, 2019, Patras, Greece.}}
  \emph{(\bibinfo{series}{LIPIcs})}, Vol.~\bibinfo{volume}{132}.
  \bibinfo{publisher}{Schloss Dagstuhl - Leibniz-Zentrum fuer Informatik},
  \bibinfo{pages}{118:1--118:13}.
\newblock


\bibitem[\protect\citeauthoryear{Kannan and Lipton}{Kannan and Lipton}{1980}]%
        {KL80}
\bibfield{author}{\bibinfo{person}{Ravindran Kannan} {and}
  \bibinfo{person}{Richard~J. Lipton}.} \bibinfo{year}{1980}\natexlab{}.
\newblock \showarticletitle{The Orbit Problem is Decidable}. In
  \bibinfo{booktitle}{\emph{Proceedings of the 12th Annual {ACM} Symposium on
  Theory of Computing, April 28-30, 1980, Los Angeles, California, {USA}}}.
  \bibinfo{pages}{252--261}.
\newblock


\bibitem[\protect\citeauthoryear{Kannan and Lipton}{Kannan and Lipton}{1986}]%
        {KL86}
\bibfield{author}{\bibinfo{person}{Ravindran Kannan} {and}
  \bibinfo{person}{Richard~J. Lipton}.} \bibinfo{year}{1986}\natexlab{}.
\newblock \showarticletitle{Polynomial-time algorithm for the orbit problem}.
\newblock \bibinfo{journal}{\emph{J. {ACM}}} \bibinfo{volume}{33},
  \bibinfo{number}{4} (\bibinfo{year}{1986}), \bibinfo{pages}{808--821}.
\newblock


\bibitem[\protect\citeauthoryear{Kincaid, Cyphert, Breck, and Reps}{Kincaid
  et~al\mbox{.}}{2018}]%
        {KCBR18}
\bibfield{author}{\bibinfo{person}{Zachary Kincaid}, \bibinfo{person}{John
  Cyphert}, \bibinfo{person}{Jason Breck}, {and} \bibinfo{person}{Thomas~W.
  Reps}.} \bibinfo{year}{2018}\natexlab{}.
\newblock \showarticletitle{Non-linear reasoning for invariant synthesis}.
\newblock \bibinfo{journal}{\emph{{PACMPL}}} \bibinfo{volume}{2},
  \bibinfo{number}{{POPL}} (\bibinfo{year}{2018}),
  \bibinfo{pages}{54:1--54:33}.
\newblock


\bibitem[\protect\citeauthoryear{Macintyre and Wilkie}{Macintyre and
  Wilkie}{1996}]%
        {MacintyreWilkie1996}
\bibfield{author}{\bibinfo{person}{Angus Macintyre} {and}
  \bibinfo{person}{Alex~J. Wilkie}.} \bibinfo{year}{1996}\natexlab{}.
\newblock \showarticletitle{On the Decidability of the Real Exponential Field}.
\newblock In \bibinfo{booktitle}{\emph{Kreiseliana. About and Around Georg
  Kreisel}}, \bibfield{editor}{\bibinfo{person}{Piergiorgio Odifreddi}} (Ed.).
  \bibinfo{publisher}{A K Peters}, \bibinfo{pages}{441--467}.
\newblock


\bibitem[\protect\citeauthoryear{Masser}{Masser}{1988}]%
        {Mas88}
\bibfield{author}{\bibinfo{person}{David~W Masser}.}
  \bibinfo{year}{1988}\natexlab{}.
\newblock \showarticletitle{Linear relations on algebraic groups}.
\newblock \bibinfo{journal}{\emph{New Advances in Transcendence Theory}}
  (\bibinfo{year}{1988}), \bibinfo{pages}{248--262}.
\newblock


\bibitem[\protect\citeauthoryear{Mignotte, Shorey, and Tijdeman}{Mignotte
  et~al\mbox{.}}{1984}]%
        {MST84}
\bibfield{author}{\bibinfo{person}{M. Mignotte}, \bibinfo{person}{T. Shorey},
  {and} \bibinfo{person}{R. Tijdeman}.} \bibinfo{year}{1984}\natexlab{}.
\newblock \showarticletitle{The distance between terms of an algebraic
  recurrence sequence}.
\newblock \bibinfo{journal}{\emph{J. f\"ur die reine und angewandte Math.}}
  \bibinfo{volume}{349} (\bibinfo{year}{1984}).
\newblock


\bibitem[\protect\citeauthoryear{Ouaknine, Pinto, and Worrell}{Ouaknine
  et~al\mbox{.}}{2015}]%
        {OPW15}
\bibfield{author}{\bibinfo{person}{Jo{\"{e}}l Ouaknine},
  \bibinfo{person}{Jo{\~{a}}o~Sousa Pinto}, {and} \bibinfo{person}{James
  Worrell}.} \bibinfo{year}{2015}\natexlab{}.
\newblock \showarticletitle{On Termination of Integer Linear Loops}. In
  \bibinfo{booktitle}{\emph{Proceedings of the Twenty-Sixth Annual {ACM-SIAM}
  Symposium on Discrete Algorithms, {SODA} 2015, San Diego, CA, USA, January
  4-6, 2015}}. \bibinfo{pages}{957--969}.
\newblock


\bibitem[\protect\citeauthoryear{Ouaknine and Worrell}{Ouaknine and
  Worrell}{2014a}]%
        {OW14b}
\bibfield{author}{\bibinfo{person}{Jo{\"{e}}l Ouaknine} {and}
  \bibinfo{person}{James Worrell}.} \bibinfo{year}{2014}\natexlab{a}.
\newblock \showarticletitle{On the Positivity Problem for Simple Linear
  Recurrence Sequences,}. In \bibinfo{booktitle}{\emph{Automata, Languages, and
  Programming - 41st International Colloquium, {ICALP} 2014, Copenhagen,
  Denmark, July 8-11, 2014, Proceedings, Part {II}}}.
  \bibinfo{pages}{318--329}.
\newblock


\bibitem[\protect\citeauthoryear{Ouaknine and Worrell}{Ouaknine and
  Worrell}{2014b}]%
        {OW14a}
\bibfield{author}{\bibinfo{person}{Jo{\"{e}}l Ouaknine} {and}
  \bibinfo{person}{James Worrell}.} \bibinfo{year}{2014}\natexlab{b}.
\newblock \showarticletitle{Positivity Problems for Low-Order Linear Recurrence
  Sequences}. In \bibinfo{booktitle}{\emph{Proceedings of the Twenty-Fifth
  Annual {ACM-SIAM} Symposium on Discrete Algorithms, {SODA} 2014, Portland,
  Oregon, USA, January 5-7, 2014}}. \bibinfo{pages}{366--379}.
\newblock


\bibitem[\protect\citeauthoryear{Ouaknine and Worrell}{Ouaknine and
  Worrell}{2014c}]%
        {ouaknine2014ultimate}
\bibfield{author}{\bibinfo{person}{Jo{\"e}l Ouaknine} {and}
  \bibinfo{person}{James Worrell}.} \bibinfo{year}{2014}\natexlab{c}.
\newblock \showarticletitle{Ultimate Positivity is decidable for simple linear
  recurrence sequences}. In \bibinfo{booktitle}{\emph{Automata, Languages, and
  Programming - 41st International Colloquium, {ICALP} 2014, Copenhagen,
  Denmark, July 8-11, 2014, Proceedings, Part {II}}}. Springer,
  \bibinfo{pages}{330--341}.
\newblock


\bibitem[\protect\citeauthoryear{Ouaknine and Worrell}{Ouaknine and
  Worrell}{2015}]%
        {OW15}
\bibfield{author}{\bibinfo{person}{Jo{\"{e}}l Ouaknine} {and}
  \bibinfo{person}{James Worrell}.} \bibinfo{year}{2015}\natexlab{}.
\newblock \showarticletitle{On linear recurrence sequences and loop
  termination}.
\newblock \bibinfo{journal}{\emph{{SIGLOG} News}} \bibinfo{volume}{2},
  \bibinfo{number}{2} (\bibinfo{year}{2015}), \bibinfo{pages}{4--13}.
\newblock


\bibitem[\protect\citeauthoryear{Rodr{\'{\i}}guez{-}Carbonell and
  Kapur}{Rodr{\'{\i}}guez{-}Carbonell and Kapur}{2004}]%
        {RK04}
\bibfield{author}{\bibinfo{person}{Enric Rodr{\'{\i}}guez{-}Carbonell} {and}
  \bibinfo{person}{Deepak Kapur}.} \bibinfo{year}{2004}\natexlab{}.
\newblock \showarticletitle{An Abstract Interpretation Approach for Automatic
  Generation of Polynomial Invariants}. In \bibinfo{booktitle}{\emph{Static
  Analysis, 11th International Symposium, {SAS} 2004, Verona, Italy, August
  26-28, 2004, Proceedings}}. \bibinfo{pages}{280--295}.
\newblock


\bibitem[\protect\citeauthoryear{Rodr{\'{\i}}guez{-}Carbonell and
  Kapur}{Rodr{\'{\i}}guez{-}Carbonell and Kapur}{2007}]%
        {RK07}
\bibfield{author}{\bibinfo{person}{Enric Rodr{\'{\i}}guez{-}Carbonell} {and}
  \bibinfo{person}{Deepak Kapur}.} \bibinfo{year}{2007}\natexlab{}.
\newblock \showarticletitle{Generating all polynomial invariants in simple
  loops}.
\newblock \bibinfo{journal}{\emph{J. Symb. Comput.}} \bibinfo{volume}{42},
  \bibinfo{number}{4} (\bibinfo{year}{2007}), \bibinfo{pages}{443--476}.
\newblock


\bibitem[\protect\citeauthoryear{Sankaranarayanan, Sipma, and
  Manna}{Sankaranarayanan et~al\mbox{.}}{2004}]%
        {SSM04}
\bibfield{author}{\bibinfo{person}{Sriram Sankaranarayanan},
  \bibinfo{person}{Henny Sipma}, {and} \bibinfo{person}{Zohar Manna}.}
  \bibinfo{year}{2004}\natexlab{}.
\newblock \showarticletitle{Non-linear loop invariant generation using
  Gr{\"{o}}bner bases}. In \bibinfo{booktitle}{\emph{Proceedings of the 31st
  {ACM} {SIGPLAN-SIGACT} Symposium on Principles of Programming Languages,
  {POPL} 2004, Venice, Italy, January 14-16, 2004}}. \bibinfo{pages}{318--329}.
\newblock


\bibitem[\protect\citeauthoryear{Tao}{Tao}{2008}]%
        {Tao08}
\bibfield{author}{\bibinfo{person}{T. Tao}.} \bibinfo{year}{2008}\natexlab{}.
\newblock \bibinfo{booktitle}{\emph{Structure and randomness: pages from year
  one of a mathematical blog}}.
\newblock \bibinfo{publisher}{American Mathematical Society}.
\newblock


\bibitem[\protect\citeauthoryear{Tarski}{Tarski}{1951}]%
        {tarski1951decision}
\bibfield{author}{\bibinfo{person}{Alfred Tarski}.}
  \bibinfo{year}{1951}\natexlab{}.
\newblock \showarticletitle{A decision method for elementary algebra and
  geometry}.
\newblock \bibinfo{journal}{\emph{RAND Corporation, R-109}}
  (\bibinfo{year}{1951}).
\newblock


\bibitem[\protect\citeauthoryear{Tiwari}{Tiwari}{2004}]%
        {Tiw04}
\bibfield{author}{\bibinfo{person}{Ashish Tiwari}.}
  \bibinfo{year}{2004}\natexlab{}.
\newblock \showarticletitle{Termination of Linear Programs}. In
  \bibinfo{booktitle}{\emph{Computer Aided Verification, 16th International
  Conference, {CAV} 2004, Boston, MA, USA, July 13-17, 2004, Proceedings}}.
  \bibinfo{pages}{70--82}.
\newblock


\bibitem[\protect\citeauthoryear{Vereshchagin}{Vereshchagin}{1985}]%
        {Ver85}
\bibfield{author}{\bibinfo{person}{N.~K. Vereshchagin}.}
  \bibinfo{year}{1985}\natexlab{}.
\newblock \showarticletitle{The problem of appearance of a zero in a linear
  recurrence sequence (in {R}ussian)}.
\newblock \bibinfo{journal}{\emph{Mat. Zametki}} \bibinfo{volume}{38},
  \bibinfo{number}{2} (\bibinfo{year}{1985}), \bibinfo{pages}{609--615}.
\newblock


\bibitem[\protect\citeauthoryear{Wilkie}{Wilkie}{1996}]%
        {Wilkie96}
\bibfield{author}{\bibinfo{person}{A.~J. Wilkie}.}
  \bibinfo{year}{1996}\natexlab{}.
\newblock \showarticletitle{Model Completeness Results for Expansions of the
  Ordered Field of Real Numbers by Restricted Pfaffian Functions and the
  Exponential Function}.
\newblock \bibinfo{journal}{\emph{Journal of the American Mathematical
  Society}} \bibinfo{volume}{9}, \bibinfo{number}{4} (\bibinfo{year}{1996}),
  \bibinfo{pages}{1051--1094}.
\newblock


\bibitem[\protect\citeauthoryear{Xia and Zhang}{Xia and Zhang}{2010}]%
        {XZ10}
\bibfield{author}{\bibinfo{person}{Bican Xia} {and} \bibinfo{person}{Zhihai
  Zhang}.} \bibinfo{year}{2010}\natexlab{}.
\newblock \showarticletitle{Termination of linear programs with nonlinear
  constraints}.
\newblock \bibinfo{journal}{\emph{J. Symb. Comput.}} \bibinfo{volume}{45},
  \bibinfo{number}{11} (\bibinfo{year}{2010}), \bibinfo{pages}{1234--1249}.
\newblock


\end{thebibliography}

\end{document}